\documentclass[11pt]{article}

\usepackage{amsthm}
\usepackage{graphicx} 
\usepackage{array} 

\usepackage{amsmath, amssymb, amsfonts, verbatim}
\usepackage{hyphenat,epsfig,subcaption,multirow}
\usepackage{nicefrac}
\usepackage{paralist}

\usepackage[usenames,dvipsnames]{xcolor}
\usepackage[ruled]{algorithm2e}

\DeclareFontFamily{U}{mathx}{\hyphenchar\font45}
\DeclareFontShape{U}{mathx}{m}{n}{
      <5> <6> <7> <8> <9> <10>
      <10.95> <12> <14.4> <17.28> <20.74> <24.88>
      mathx10
      }{}
\DeclareSymbolFont{mathx}{U}{mathx}{m}{n}
\DeclareMathSymbol{\bigtimes}{1}{mathx}{"91}

\usepackage{tcolorbox}
\tcbuselibrary{skins,breakable}
\tcbset{enhanced jigsaw}

\usepackage[normalem]{ulem}
\usepackage[compact]{titlesec}

\definecolor{DarkRed}{rgb}{0.5,0.1,0.1}
\definecolor{DarkBlue}{rgb}{0.1,0.1,0.5}

\usepackage{nameref}
\definecolor{ForestGreen}{rgb}{0.1333,0.5451,0.1333}
\definecolor{Red}{rgb}{0.9,0,0}
\usepackage[linktocpage=true,
	pagebackref=true,colorlinks,
	linkcolor=DarkRed,citecolor=ForestGreen,
	bookmarks,bookmarksopen,bookmarksnumbered]
	{hyperref}
\usepackage[noabbrev,nameinlink]{cleveref}
\crefname{property}{property}{Property}
\creflabelformat{property}{(#1)#2#3}
\crefname{equation}{eq}{Eq}
\creflabelformat{equation}{(#1)#2#3}

\usepackage{bm}
\usepackage{url}
\usepackage{xspace}
\usepackage[mathscr]{euscript}

\usepackage{tikz}
\usetikzlibrary{arrows}
\usetikzlibrary{arrows.meta}
\usetikzlibrary{shapes}
\usetikzlibrary{backgrounds}
\usetikzlibrary{positioning}
\usetikzlibrary{decorations.markings}
\usetikzlibrary{patterns}
\usetikzlibrary{calc}
\usetikzlibrary{fit}
\usetikzlibrary{snakes}

\usepackage{mdframed}

\usepackage[noend]{algpseudocode}
\makeatletter
\def\BState{\State\hskip-\ALG@thistlm}
\makeatother

\usepackage{cite}
\usepackage{enumitem}

\usepackage[margin=1in]{geometry}

\newtheorem{theorem}{Theorem}
\newtheorem{lemma}{Lemma}[section]
\newtheorem{proposition}[lemma]{Proposition}
\newtheorem{corollary}[lemma]{Corollary}
\newtheorem{claim}[lemma]{Claim}

\newtheorem*{claim*}{Claim}
\newtheorem*{proposition*}{Proposition}
\newtheorem*{lemma*}{Lemma}
\newtheorem*{problem*}{Problem}

\crefname{lemma}{Lemma}{Lemmas}
\crefname{claim}{Claim}{Claims}

\newtheorem{mdresult}{Result}
\newenvironment{result}{\begin{mdframed}[backgroundcolor=lightgray!40,topline=false,rightline=false,leftline=false,bottomline=false,innertopmargin=2pt]\begin{mdresult}}{\end{mdresult}\end{mdframed}}

\theoremstyle{definition}

\newtheorem{mdinvariant}[lemma]{Definition}
\newenvironment{Definition}{\begin{mdframed}[hidealllines=false,innerleftmargin=5pt,backgroundcolor=gray!10,innertopmargin=2pt]\begin{mdinvariant}}{\end{mdinvariant}\end{mdframed}}

\newtheorem{mdalg}{Algorithm}
\newenvironment{Algorithm}{\begin{tbox}\begin{mdalg}}{\end{mdalg}\end{tbox}}

\newtheoremstyle{restate}{}{}{\itshape}{}{\bfseries}{~(restated).}{.5em}{\thmnote{#3}}
\theoremstyle{restate}

\allowdisplaybreaks

\newcommand{\Qed}[1]{\ensuremath{\qed_{\,\,\textnormal{\Cref{#1}}}}}

\renewcommand{\qed}{\nobreak \ifvmode \relax \else
      \ifdim\lastskip<1.5em \hskip-\lastskip
      \hskip1.5em plus0em minus0.5em \fi \nobreak
      \vrule height0.75em width0.5em depth0.25em\fi}

\newcommand{\eps}{\ensuremath{\varepsilon}}
\newcommand{\Paren}[1]{\Big(#1\Big)}
\newcommand{\Bracket}[1]{\Big[#1\Big]}
\newcommand{\bracket}[1]{\left[#1\right]}
\newcommand{\paren}[1]{\ensuremath{\left(#1\right)}\xspace}
\newcommand{\card}[1]{\left\vert{#1}\right\vert}

\newcommand{\IN}{\ensuremath{\mathbb{N}}}

\newcommand{\floor}[1]{{\left\lfloor{#1}\right\rfloor}}

\newcommand{\expect}[1]{\Exp\bracket{#1}}

\newcommand{\set}[1]{\ensuremath{\left\{ #1 \right\}}}
\newcommand{\poly}{\mbox{\rm poly}}

\DeclareMathOperator*{\Exp}{\ensuremath{{\mathbb{E}}}}
\DeclareMathOperator*{\Prob}{\ensuremath{\textnormal{Pr}}}
\renewcommand{\Pr}{\Prob}

\newenvironment{tbox}{\begin{tcolorbox}[
		enlarge top by=5pt,
		enlarge bottom by=5pt,
		 breakable,
		 boxsep=0pt,
                  left=4pt,
                  right=4pt,
                  top=10pt,
                  arc=0pt,
                  boxrule=1pt,toprule=1pt,
                  colback=white
                  ]
	}
{\end{tcolorbox}}

\newcommand{\event}{\ensuremath{\mathcal{E}}}

\newcommand{\supp}[1]{\ensuremath{\textnormal{\text{supp}}(#1)}}

\newcommand{\II}{\ensuremath{\mathbb{I}}}

\newcommand{\mireal}[1][]{
  \ifx\relax#1\relax%
    \II(\mione \,; \mitwo)%
  \else%
    \II(\mione \,; \mitwo\mid #1)%
  \fi
}

\newcommand{\coloring}{\ensuremath{\textnormal{\textsf{Coloring}}}\xspace}

\newcommand{\Gi}[1]{\ensuremath{G}_{#1}}
\newcommand{\Mi}[1]{\ensuremath{M}_{#1}}
\newcommand{\Mstari}[1]{\ensuremath{M^*}_{#1}}

\newcommand{\Gbase}{\ensuremath{G_{\textnormal{\texttt{Base}}}}\xspace}
\newcommand{\Ebase}{\ensuremath{E_{\textnormal{\texttt{Base}}}}\xspace}

\newcommand{\Gmiss}{\ensuremath{G_{\textnormal{\texttt{Miss}}}}\xspace}
\newcommand{\Emiss}{\ensuremath{E_{\textnormal{\texttt{Miss}}}}\xspace}

\newcommand{\PPi}[1]{P_{#1}}
\newcommand{\Ei}[1]{E_{#1}}

\newcommand{\GR}{\ensuremath{\mathbb{G}}}
\newcommand{\GRp}{\ensuremath{\widetilde{\GR}}}

\newcommand{\phistar}{\ensuremath{\phi^{*}}}

\newcommand{\CC}{\ensuremath{\mathcal{C}}}
\newcommand{\CCstar}{\ensuremath{\mathcal{C}^{\star}}}

\newcommand{\istar}{i^{\star}}

\title{Deterministic Graph Coloring in the Streaming Model}
\author{
Sepehr Assadi\footnote{(\texttt{sepehr.assadi@rutgers.edu}) Department of Computer Science, Rutgers University. Research supported in part by the NSF CAREER award CCF-2047061 and a gift from Google Research.} \and 
Andrew Chen\footnote{(\texttt{ac2337@cornell.edu}) Department of Mathematics, Cornell University. Research  done primarily as part of 2020 REU program at DIMACS and Rutgers, supported by the NSF grant CCF-1852215.} \and
Glenn Sun\footnote{(\texttt{glennsun@ucla.edu}) Department of Mathematics, University of California--Los Angeles. Research  done primarily as part of 2021 REU program at DIMACS and Rutgers, supported by the NSF grant CCF-1836666.}}

\date{}

\begin{document}
\maketitle

\pagenumbering{roman}

\begin{abstract}
	Recent breakthroughs in graph streaming have led to the design of single-pass semi-streaming algorithms for various graph coloring problems such as $(\Delta+1)$-coloring, degeneracy-coloring, coloring triangle-free graphs, and others. 
	These algorithms are all {randomized} in crucial ways and whether or not there is any deterministic analogue of them has remained an important open question in this line of work. 
	
	\medskip
	
	We  settle this fundamental question by proving that  there is no {deterministic} single-pass semi-streaming algorithm that given a graph $G$ with maximum degree $\Delta$, can output a proper coloring of $G$ 
	using any number of colors which is sub-exponential in $\Delta$. Our proof is based on analyzing the multi-party communication complexity of a related communication game, using 
	random graph theory type arguments that may be of independent interest. 
	
	\medskip
	
	We complement our  lower bound by showing that just one extra pass over the input allows one to recover an $O(\Delta^2)$ coloring via a deterministic semi-streaming algorithm. This result is further extended 
	to an $O(\Delta)$ coloring  in $O(\log{\Delta})$ passes even in dynamic streams. 
\end{abstract}

\clearpage

\setcounter{tocdepth}{3}
\tableofcontents

\clearpage

\pagenumbering{arabic}
\setcounter{page}{1}


\section{Introduction}\label{sec:intro}

Coloring graphs with a small number of colors is a central problem in graph theory with a wide range of applications in computer science. A proper $c$-coloring of a graph $G=(V,E)$ assigns a color from the palette $\set{1,\ldots,c}$ 
to the vertices so that no edge is monochromatic. 
We study  graph coloring  in the semi-streaming model introduced by~\cite{FeigenbaumKMSZ05}: 
the edges of an $n$-vertex input graph are arriving one by one in a stream and the algorithm can make one (or a few) passes over the stream and use a limited memory of $O(n \cdot \poly\!\log\!{(n)})$ bits. At the end, 
it should output a proper coloring of the input graph. The semi-streaming model is particularly motivated by its applications to processing massive graphs and has received extensive attention in the last two decades.  

Similar to the classical setting, it is known that approximating the minimum number of colors for proper coloring is quite intractable in the semi-streaming model~\cite{HalldorssonSSW12,AbboudCKP19,CormodeDK19}. 
As a result, the  interest in this problem in graph streaming has primarily been on obtaining colorings with number of colors proportional to certain combinatorial parameters of input graphs, such as maximum degree or degeneracy. 
On this front, a breakthrough result of~\cite{AssadiCK19a} gave the first semi-streaming algorithm for $(\Delta+1)$ coloring of graphs with maximum degree $\Delta$ (see also the independent work of~\cite{BeraG18} that obtained
an $O(\Delta)$ coloring algorithm). Another remarkable result is that of~\cite{BeraCG19} that gave a semi-streaming algorithm for $(\kappa+o(\kappa))$-coloring of graphs with degeneracy $\kappa$. See~\cite{BehnezhadDHKS19,CormodeDK19,AlonA20,BhattacharyaBMU21} for other related results. 

Perhaps, the single most common characteristic of all results in this line of work is that \emph{they crucially rely on randomization}. For instance, one of the strongest tool for streaming graph coloring is the {palette sparsification theorem} of~\cite{AssadiCK19a} 
which states the following: if we sample $O(\log{n})$ colors from $\set{1,\ldots,\Delta+1}$ for each vertex
independently and uniformly at random, then with high probability, the entire graph can be colored using only the sampled colors of each vertex. 
This result immediately leads to a semi-streaming algorithm for $(\Delta+1)$ coloring: after sampling $O(\log{n})$ colors for each vertex, only $O(n\log^2\!{(n)})$ edges can potentially become monochromatic under \emph{any} coloring of vertices from their sampled colors; thus, the algorithm can simply store these edges 
throughout the stream and find the desired coloring at the end (which is guaranteed to exist by the palette sparsification theorem). But the resulting algorithm is inherently randomized with this tool. 

This state-of-affairs of graph coloring  in admitting only randomized semi-streaming algorithms is rather unusual in the literature. Indeed, most  problems of interest 
in the semi-streaming model such as (minimum) spanning trees~\cite{FeigenbaumKMSZ05}, edge/vertex connectivity~\cite{GuhaMT15}, cut and spectral sparsifiers~\cite{McGregor14}, spanners~\cite{FeigenbaumKMSZ05,FeigenbaumKMSZ08} and weighted matchings~\cite{PazS17} all admit deterministic algorithms with the same performance as best known randomized algorithms\footnote{There are some other exceptions to this rule also; moreover, in many cases, randomization can further help, e.g., by 
reducing the runtime of algorithms, but typically not that much with their space. We also emphasize that this ``rough equivalence of power'' of deterministic vs randomized algorithms only exist in the semi-streaming model: once 
we reduce the space to $o(n)$, deterministic algorithms are  much weaker than randomized ones for most problems.} (or altogether do not  admit  non-trivial randomized algorithms; see,
 e.g.~\cite{FeigenbaumKMSZ08,AssadiKL16,AssadiCK19a,CormodeDK19,BhattacharyaBMU21} for various examples of such impossibility results). Consequently, there has been a general interest in de-randomizing the semi-streaming algorithms for graph coloring, following the same recent trend in various closely related models such as distributed
computing~\cite{Parter18,CzumajDP20,Censor-HillelPS20,GhaffariK20} and Massively Parallel Computation (MPC) algorithms~\cite{CzumajDP21a,CzumajDP21b}. This has led to the following 
important open question:
\vspace{-5pt}
\begin{quote}
	\emph{Can we design \emph{\textbf{deterministic semi-streaming algorithms}} for graph coloring with similar guarantees as the randomized ones? In particular, are there deterministic semi-streaming algorithms for $(\Delta+1)$-coloring, $O(\Delta)$ coloring, 
	or even $\poly{(\Delta)}$ coloring? } 
\end{quote} 
\vspace{5pt}

\bigskip


\subsection{Our Contributions}\label{sec:results}

Our main result is a strong negative answer to this fundamental open question: coloring graphs even with $\exp\paren{\Delta^{o(1)}}$ colors is not possible with a deterministic semi-streaming algorithm! 

\begin{result}\label{res:lower}
	There does \underline{not} exist any deterministic single-pass semi-streaming algorithm for coloring graphs of maximum degree $\Delta$ using at most $\exp(\Delta^{o(1)})$ colors (even when $\Delta$ is known to the algorithm at the beginning of the stream).
\end{result}

We note that~\Cref{res:lower} extends to the entire range of streaming algorithms with $o(n\Delta)$ space as well; see~\Cref{cor:main-lower} for the formalization of this result and precise bounds.   

Previously, no space lower bound was known for deterministic semi-streaming algorithms even for $(\Delta+1)$-coloring and even for dynamic streams that also allow for deleting edges from the 
stream\footnote{Unlike insertion-only streams, \emph{all} known algorithms in dynamic streams are randomized and for a crucial reason. 
It is easy to see that any non-trivial algorithm that should  return a single edge from the graph cannot be deterministic in dynamic streams: one can simply use the memory of the algorithm to recover the entire input  by passing each returned edge as a deletion to the algorithm, hence forcing it to return another edge of the graph, until we recover the entire graph. This means the memory of the algorithm has to be $\Omega(n^2)$ bits, enough to store the entire input. This approach however does not 
apply to $(\Delta+1)$-coloring at it does not require returning any edge as output.} (but see~\Cref{sec:recent} for a recent independent work). On the other hand,~\Cref{res:lower}
effectively rules out  any non-trivial algorithm for graph coloring: the best thing to do in $O(n\log^{q}\!{(n)})$ space is to either store the entire input graph when $\Delta \lesssim \log^{q}{(n)}$ and find a $(\Delta+1)$ coloring at the end, or color 
all vertices differently which results in $n \approx \exp\paren{\Delta^{1/q}}$-coloring for $\Delta \gtrsim \log^{q}{(n)}$. Combined with the randomized algorithm of~\cite{AssadiCK19a} for $(\Delta+1)$ coloring,~\Cref{res:lower} presents one of the strongest separations between deterministic and randomized algorithms in the semi-streaming model. 

Given the strong impossibility result of~\Cref{res:lower}, it is natural to consider standard relaxation of the problem. For this, we consider \emph{multi-pass} algorithms that read the stream more than once. 
Multi-pass algorithms have also been studied extensively since the introduction of semi-streaming algorithms in~\cite{FeigenbaumKMSZ05}. We show that unlike in a single pass, deterministic semi-streaming 
multi-pass algorithms can indeed solve non-trivial graph coloring problems already in just two passes. 

\begin{result}\label{res:upper}
	There exist deterministic semi-streaming algorithms for coloring graphs of maximum degree $\Delta$ using $O(\Delta^2)$ colors in two passes or $O(\Delta)$ colors in $O(\log{\Delta})$ passes. 
	The algorithms can be implemented even in dynamic streams with edge deletions (still deterministically). 
\end{result}

Previously, no non-trivial deterministic semi-streaming algorithm was known for graph coloring. In light of~\Cref{res:lower}, our algorithms in~\Cref{res:upper} also provide one of the strongest separation between two-pass and single-pass algorithms 
(see~\cite{AssadiD21} for another example via min-cuts). Finally, our algorithms in~\Cref{res:upper} are among the first deterministic algorithms that work on dynamic streams.   

\medskip
All in all, our results collectively establish surprising aspects of graph coloring in the semi-streaming model, further cementing the role of this fundamental problem in capturing various different separations and properties in this model. 

\subsection{Our Techniques}\label{sec:techniques} 

We now give a quick summary of our techniques here. More details can be found in the high-level overview of our approach in~\Cref{sec:overview}. 

\paragraph{Lower bound of~\Cref{res:lower}.} Our lower bound in~\Cref{res:lower} is proven by considering the \emph{multi-party} communication complexity of the coloring problem: here, the edges of input graph are partitioned across the players and they can speak in turn, once each, to compute 
a proper coloring of the input using as small as possible number of colors. It is a standard fact that communication complexity lower bounds the space of streaming algorithms. The main technical contribution of 
our work is thus a communication lower bound  for this problem. 

We obtain our lower bound by designing an adversary that specifies the inputs of players via \emph{random subgraphs} chosen \emph{adaptively} based on the messages of prior players. The adaptivity in distribution of  inputs allows us to prove a lower bound  specifically for deterministic algorithms (as a non-adaptive distributional lower bound 
also works for randomized algorithms by Yao's minimax principle~\cite{Yao77}). At the same time, working with these distributional inputs makes our arguments 
much simpler compared to using a typical counting argument over all possible graphs (we elaborate more on this in~\Cref{sec:overview}). One main ingredient of this proof is determining the power of 
communication protocols for ``compressing \emph{non-edges}'' in a random subgraph, compared to standard approaches that bound the number of \emph{edges} that can be recovered from a compression.  

\paragraph{Algorithms of~\Cref{res:upper}.} Our algorithmic results are based on finding a way to \emph{non-properly} color the graph using a small number of colors, so that 
the number of monochromatic edges is small. We can then store these edges explicitly and use them to further refine this non-proper coloring to a proper coloring of the entire graph (for $O(\Delta^2)$ coloring) 
or further extending a partial coloring and recurse (for $O(\Delta)$ coloring).

To be able to implement this strategy, we design families of coloring functions of small size so that for any given graph, at least one of these coloring functions lead to the 
desired non-proper coloring with a small number of monochromatic edges. These families are obtained via  standard tools in de-randomization, namely,  {\emph{near-universal hash functions}}.  

\subsection{Recent Related Work}\label{sec:recent} 
Independently and concurrently to us,~\cite{ChakrabartiGS21}  studied graph coloring in the semi-streaming model but for \emph{adversarially robust} algorithms (see~\cite{ChakrabartiGS21} and~\cite{Ben-EliezerJWY20} for definition and context). They prove that no semi-streaming algorithm 
can be adversarially robust when using $o(\Delta^2)$ coloring. As all deterministic algorithms are adversarially robust, their result also implies that no deterministic semi-streaming algorithm can achieve an $o(\Delta^2)$ coloring. The authors of~\cite{ChakrabartiGS21} 
also state that: ``A major remaining open question is whether this [lower bound] can be matched, perhaps by a deterministic semi-streaming $O(\Delta^2)$ coloring algorithm. In fact,
it is not known how to get even a $\poly{(\Delta)}$-coloring deterministically''. Our~\Cref{res:lower} fully settles their open question for deterministic algorithms in negative. Incidentally,~\cite{ChakrabartiGS21}  provides a randomized but adversarially robust 
semi-streaming algorithm for $O(\Delta^3)$ coloring. Thus that one cannot hope for an $\exp(\Delta^{o(1)})$ coloring lower bound like ours in their model. Technique-wise, the two work are  entirely disjoint. 

\subsection{Further Related Work}\label{sec:related}
Recently, there has been a surge of interest in graph coloring and related problems in graph streams
~\cite{HalldorssonHLS16,CormodeDK18,AssadiCK19a,BehnezhadDHKS19,CormodeDK19,KonradPR019,BeraCG19,AlonA20,BhattacharyaBMU21}. Beside what already mentioned, another work related to ours is~\cite{AlonA20} that studied  
graph theoretic aspects of palette sparsification theorem of~\cite{AssadiCK19a} and obtained semi-streaming algorithms for coloring triangle-free graphs and $(\deg+1)$-coloring. Moreover,~\cite{BhattacharyaBMU21}  showed that 
some of the ``easiest'' problems in coloring are still intractable in the semi-streaming model (even with randomization). See also~\cite{McGregor14} for an excellent overview of work on other problems in the semi-streaming model.

        
        \definecolor{lightblue}{RGB}{207, 229, 255}
\definecolor{lightred}{RGB}{255, 207, 207}
\definecolor{lightgray}{RGB}{240, 240, 240}

\section{High-Level Overview}\label{sec:overview} 

We give a streamlined overview of our approach in this section.  We emphasize that
this section oversimplifies many details and the discussions will be informal for the sake of intuition.

\subsection{Lower Bound of~\Cref{res:lower}} 
As stated earlier, the proof of~\Cref{res:lower} is by considering the multi-party communication complexity of the coloring problem. To start, let us consider the simple case of \emph{two} players Alice and Bob, 
receiving edges of a graph $G$ with maximum degree $\Delta$. Alice sends a message $M$ to Bob and Bob outputs a proper coloring of $G$ using as small as possible number of colors. What is the best strategy of players for solving the problem 
with limited communication and small number of colors? 

Before getting to this question, let us make an important remark: Coloring with more than $\Delta$ colors is inherently a \emph{search} problem not a \emph{decision} one as all graphs can be colored with $\Delta+1$ colors after all. Thus, the above question is basically asking how much Bob should learn about Alice's input to agree on a proper coloring of the \emph{entire} graph (without knowing all edges of Alice). This view will be important throughout this discussion and our formal lower bound arguments. 

\paragraph{Two-player communication complexity of coloring.} There is a simple solution to our two-player communication game using $\approx n$ size messages and $O(\Delta^2)$ colors. Alice simply sends a $(\Delta+1)$ coloring of her input graph to Bob and Bob further finds a $(\Delta+1)$ coloring of 
each of Alice's color classes \emph{individually} to obtain a proper $(\Delta+1)^2$ coloring of the entire input graph. Let us show that this is essentially the best one can do using $O(n)$ size messages and for a specific 
choice of $\Delta = \Theta(\sqrt{n})$ (neither of these assumptions are needed in our main lower bound).  

Suppose Alice receives an arbitrary graph with maximum degree $\sqrt{n}$ and maps it to a message of size $O(n)$. As the graphs with maximum degree $\sqrt{n}$ are a 
\emph{constant fraction} of graphs with $(n^{3/2}/2)$ edges, we have that there is a message, to which, Alice is mapping at least 
\[
	\Omega(1) \cdot {{{n}\choose{2}}\choose{\frac{n^{3/2}}{2}}} \cdot 2^{-O(n)} \gtrsim \exp\paren{\dfrac{n^{3/2}}{4} \cdot \ln{(n)}} \cdot 2^{-O(n)}, 
\]
many different graphs. At the same time, given this message, Bob should avoid coloring any pairs of vertices the same if they appear in \emph{some} graph mapped to this message. 
But having so many graphs mapped to the same message only allows for $O(n^{3/2})$ pairs of vertices to not have any edge at all in \emph{any} of these graphs; this is because the total number of graphs 
with maximum degree $\sqrt{n}$ whose edges avoid a fixed set of $O(n^{3/2})$ pairs of vertices have size at most
\[
	{{{{n}\choose{2}}-O(n^{3/2})}\choose{\frac{n^{3/2}}{2}}} \lesssim {{{n}\choose{2}}\choose{\frac{n^{3/2}}{2}}} \cdot (1-\frac{1}{O(\sqrt{n})})^{(\frac{n^{3/2}}{2})}  \lesssim  \exp\paren{\dfrac{n^{3/2}}{4} \cdot \ln{(n)}} \cdot 2^{-O(n)}.
\]

At this point, this means that \emph{from the perspective of Bob}, only $O(n^{3/2})$ pairs of vertices can be colored the same, even ignoring his own input graph (see~\Cref{fig:two-player} for an illustration). Moreover, a Markov bound implies that half the 
vertices only have $O(\sqrt{n})$ non-edges from the perspective of Bob. Thus, Bob will ``see'' a set $S$ of $\Theta(n)$ vertices where each one has at most $O(\sqrt{n})$ non-edges inside $S$. But recall that we 
are considering the case where maximum degree can be as large as $\Theta(\sqrt{n})$. So Bob's own input can simply contain all non-edges inside $S$ while keeping the maximum degree of the graph still $O(\sqrt{n})$. 
At this point, the induced subgraph on vertices $S$, {from the perspective of Bob}, is simply a clique, and thus requires $\card{S} = \Omega(n)$ colors. 
Since $\Delta = \Theta(\sqrt{n})$, this gives us an $\Omega(\Delta^2)$ lower bound on the number of colors. 

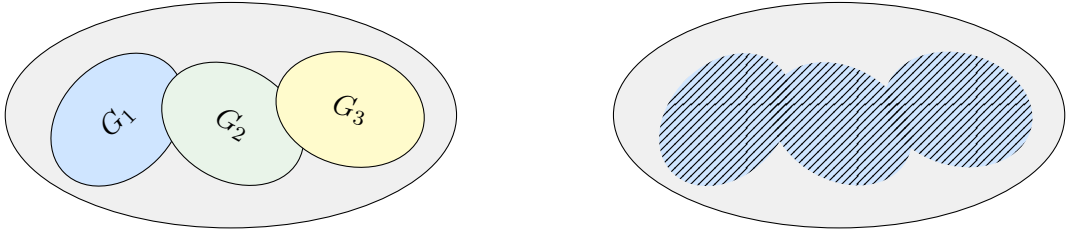
\begin{figure}[t!]
	\begin{center}
	\subcaptionbox{Alice has to map several graphs to the same message. These graphs are \emph{individually} ``sparse'': they have \emph{max}-degree $\lesssim \sqrt{n}$. \label{fig:Alice}}%
 	 [.45\linewidth]{
 		\begin{tikzpicture}
			\node[ellipse, draw, minimum width=6cm, minimum height=3cm, fill=lightgray] (E) {}; 
			\node[ellipse, draw, rotate=45, minimum width=2cm, minimum height=1.5cm, fill=lightblue] (E1) [above left=-1cm and -1.5cm of E] {$G_1$}; 
			\node[ellipse, draw, rotate=-30, minimum width=2cm, minimum height=1.5cm, fill=ForestGreen!10] (E2) [above left=-2cm and -2.5cm of E] {$G_2$}; 
			\node[ellipse, draw, rotate=-15, minimum width=2cm, minimum height=1.5cm, fill=Yellow!25] (E3) [above right=-0.25cm and -0.15cm of E2] {$G_3$}; 
		\end{tikzpicture}
	}
  \hspace{0.4cm} 
  \subcaptionbox{Bob however ``sees'' all these edges as part of the input. So, from Bob's perspective, this subgraph is ``dense'': it has \emph{min}-degree $\gtrsim n-\sqrt{n}$. Thus, even a ``sparse'' input to Bob with max-degree $\lesssim \sqrt{n}$, 
  turns this subgraph into a clique.   \label{fig:Bob}}%
  [.45\linewidth]{
 		\begin{tikzpicture}
			\node[ellipse, draw, minimum width=6cm, minimum height=3cm, fill=lightgray] (E) {}; 

			\node[ellipse, draw=none, rotate=45, minimum width=2cm, minimum height=1.5cm, fill=lightblue] (E1) [above left=-1cm and -1.5cm of E] {}; 
			\node[ellipse, draw=none, rotate=-30, minimum width=2cm, minimum height=1.5cm, fill=lightblue] (E2) [above left=-2cm and -2.5cm of E] {}; 
			\node[ellipse, draw=none, rotate=-15, minimum width=2cm, minimum height=1.5cm, fill=lightblue] (E3) [above right=-0.25cm and -0.15cm of E2] {}; 
			
			\node[ellipse, draw=none, rotate=45, minimum width=2cm, minimum height=1.5cm, fill=black!25, pattern=north east lines] (E1) [above left=-1cm and -1.5cm of E] {}; 
			\node[ellipse, draw=none, rotate=-30, minimum width=2cm, minimum height=1.5cm, fill=ForestGreen!10, pattern=north east lines] (E2) [above left=-2cm and -2.5cm of E] {}; 
			\node[ellipse, draw=none, rotate=-15, minimum width=2cm, minimum height=1.5cm, fill=Yellow!25, pattern=north east lines] (E3) [above right=-0.25cm and -0.15cm of E2] {}; 
		\end{tikzpicture}
  }
	\end{center}
	\caption{An illustration of the two-player communication lower bound.}
	\label{fig:two-player}
\end{figure}

\paragraph{Multi-party communication complexity of coloring.} Given the protocol mentioned earlier for two players, to prove~\Cref{res:lower}, we need to consider a larger number of players. In general, the same strategy outlined above also implies a protocol for $k$ players with $O(n)$ communication per player and an $O(\Delta^{k})$ coloring. 
Our goal is to match this in our lower bound. 

Suppose now we have $k$ players $\PPi{1},\ldots,\PPi{k}$ and the input edges are partitioned between them. Let us again present a graph of maximum degree $\approx \Delta/k$ to the first player. We can again 
use a similar counting argument to bound the number of non-edges in inputs mapped to a message of player $\PPi{1}$ (assuming that it has size, say, $O(n)$).  
We would like to continue this procedure, by choosing the input graph of player $\PPi{2}$ in a way that ``destroys'' many of these pairs, while having maximum degree of still $\approx \Delta/k$; then recourse on the third player and so on. However, 
continuing the above counting argument directly seems intractable at this point. 

\begin{figure}[t!]
	\begin{center}
	\subcaptionbox{Player $\PPi{i}$'s different inputs that are mapped to the same message. The right (white) part are the vertices already removed from consideration and the left (dark) part are the 
	``dense'' subgraph of the input from the perspective of $\PPi{i}$. \label{fig:11}}%
 	 [.45\linewidth]{
 		\begin{tikzpicture}
			\node[rectangle, draw, minimum width=6cm, minimum height=3cm, fill=lightgray] (E) {}; 
			\node[rectangle, draw, minimum width=1cm, minimum height=2.7cm, fill=white] (E1) [above right=-2.85cm and -1.2cm of E]{}; 
			\node[rectangle, draw, minimum width=4.5cm, minimum height=2.7cm, fill=black!20] (E2) [left= 0.1cm of E1]{}; 
			
			\node[ellipse, draw, rotate=45, minimum width=1.75cm, minimum height=1.25cm, fill=lightblue] (X1) [above left=-1cm and -1.75cm of E2] {$G_1$}; 
			\node[ellipse, draw, rotate=-30, minimum width=1.75cm, minimum height=1.25cm, fill=ForestGreen!10] (X2) [above left=-2.5cm and -2.5cm of E2] {$G_2$}; 
			\node[ellipse, draw, rotate=15, minimum width=1.75cm, minimum height=1.25cm, fill=Yellow!25] (X3) [above right=0.25cm and -0.15cm of X2] {$G_3$}; 
		\end{tikzpicture}
	}
  \hspace{0.4cm} 
  	\subcaptionbox{For player $\PPi{i+1}$, the left (dark) part ``looks'' even more ``dense'' than it was for player $\PPi{i}$, as multiple different graphs of $\PPi{i}$'s input are mapped to the same message. \label{fig:22}}%
 	 [.45\linewidth]{
 		\begin{tikzpicture}
			\node[rectangle, draw, minimum width=6cm, minimum height=3cm, fill=lightgray] (E) {}; 
			\node[rectangle, draw, minimum width=1cm, minimum height=2.7cm, fill=white] (E1) [above right=-2.85cm and -1.2cm of E]{}; 
			\node[rectangle, draw, minimum width=4.5cm, minimum height=2.7cm, fill=black!20] (E2) [left= 0.1cm of E1]{}; 
			
			\node[ellipse, rotate=45, minimum width=1.75cm, minimum height=1.25cm, pattern=north east lines] (X1) [above left=-1cm and -1.75cm of E2] {}; 
			\node[ellipse, rotate=-30, minimum width=1.75cm, minimum height=1.25cm, pattern=north east lines] (X2) [above left=-2.5cm and -2.5cm of E2] {}; 
			\node[ellipse, rotate=15, minimum width=1.75cm, minimum height=1.25cm,  pattern=north east lines] (X3) [above right=0.25cm and -0.15cm of X2] {}; 
		\end{tikzpicture}
	}
	
  \vspace{1cm} 
  	\subcaptionbox{We further remove ``less dense'' part of the input (middle layer) and provide the inputs of $\PPi{i+1}$ inside the remaining subgraph. \label{fig:33}}%
 	 [.45\linewidth]{
  		\begin{tikzpicture}
			\node[rectangle, draw, minimum width=6cm, minimum height=3cm, fill=lightgray] (E) {}; 
			\node[rectangle, draw, minimum width=1cm, minimum height=2.7cm, fill=white] (E1) [above right=-2.85cm and -1.2cm of E]{}; 
			\node[rectangle, draw, minimum width=1cm, minimum height=2.7cm, fill=black!20] (E2) [left= 0.1cm of E1]{}; 
			\node[rectangle, draw, minimum width=3.4cm, minimum height=2.7cm, fill=black!35] (E3) [left= 0.1cm of E2]{}; 		
		
				\node[ellipse, draw, rotate=45, minimum width=1.75cm, minimum height=1.25cm, fill=lightblue] (X1) [above left=-1cm and -1.75cm of E3] {$G_1$}; 
			\node[ellipse, draw, rotate=0, minimum width=1.75cm, minimum height=1.25cm, fill=ForestGreen!10] (X2) [above left=-2.5cm and -2.5cm of E3] {$G_2$}; 
			\node[ellipse, draw, rotate=15, minimum width=1.75cm, minimum height=1.25cm, fill=Yellow!25] (X3) [above right=0.25cm and -0.5cm of X2] {$G_3$}; 
	
			\end{tikzpicture}
	}
  \hspace{0.4cm} 
  \subcaptionbox{We continue like this until the last player; at that point, the remaining ``super dense'' part of the input (left most part) from the perspective of $\PPi{k}$ is simply a clique.   \label{fig:44}}%
  [.45\linewidth]{
 		\begin{tikzpicture}
  	
			\node[rectangle, draw, minimum width=6cm, minimum height=3cm, fill=lightgray] (E) {}; 
			\node[rectangle, draw, minimum width=1cm, minimum height=2.7cm, fill=white] (E1) [above right=-2.85cm and -1.2cm of E]{}; 
			\node[rectangle, draw, minimum width=1cm, minimum height=2.7cm, fill=black!20] (E2) [left= 0.1cm of E1]{}; 
			\node[rectangle, draw, minimum width=1cm, minimum height=2.7cm, fill=black!35] (E3) [left= 0.1cm of E2]{}; 		
			\node[rectangle, draw, minimum width=1cm, minimum height=2.7cm, fill=black!50] (E4) [left= 0.1cm of E3]{}; 		
			\node[rectangle, draw, minimum width=1.2cm, minimum height=2.7cm, fill=black!55] (E5) [left= 0.1cm of E4]{}; 		
			\node[rectangle, draw, minimum width=1.2cm, minimum height=2.7cm, pattern=north east lines] (E6) [left= 0.1cm of E4]{}; 		
		
			\end{tikzpicture}
  }
	\end{center}
	\caption{An illustration of the multi-player communication lower bound.}
	\label{fig:multi-player}
\end{figure}

It turns out however that there is an easier way to implement this strategy by providing the input of players as \emph{random subgraphs}. Specifically, the process goes as follows (see~\Cref{fig:multi-player}):
\begin{itemize}[leftmargin=10pt]
	\item We present the first player $\PPi{1}$ with a random Erd\H{o}s-R\'enyi graph with probability $\approx (\Delta/kn)$ for each edge (so max-degree $\approx \Delta/k$ with high probability). 
	We prove that (see our \textbf{Compression Lemma} below) that there is some message $M_1$ of $\PPi{1}$ that creates $\lesssim k \cdot n^2/\Delta$ non-edges from the perspective of remaining players. 
	We further \emph{remove} all vertices with non-edge-degree $\gtrsim k^2 n/\Delta$ which by Markov bound are only $\lesssim n/k$. 
	
	\item To player $\PPi{2}$, we give a \emph{random subgraph} of (remaining) non-edges left by $M_1$ where each edge appears with probability $\approx (\Delta^2/k^3n)$ now. 
	By the bound of $\lesssim k^2 n/\Delta$ on the non-edge-degree of remaining vertices, it is easy to see that the input given to $\PPi{2}$ \emph{still} has max-degree $\approx \Delta/k$ with high probability. We again
	use the \textbf{Compression Lemma} to find a message $M_2$ of $\PPi{2}$ that creates $\lesssim k^3  n^2/\Delta^2$ non-edges from the perspective of subsequent players, and continue. 
	This way, each step to the next player will remove $\lesssim n/k$ vertices while reducing non-edge-degree of remaining vertices by a $\gtrsim \Delta/k^2$ factor. 
	
	\item Eventually, we will be able to give a {random subgraph} of non-edges left by $M_1,\ldots,M_{k-1}$ to the player $\PPi{k}$ with edge  probability $\approx (\Delta^k/k^{2k} n)$, and bound 
	the \emph{total} maximum-degree of the graph by $k \cdot \Delta/k = \Delta$ as desired. 
	But if we assume that $(\Delta/k^2)^k \approx n$ (again, this assumption is only for simplicity of exposition here), 
	it means that we turned the \emph{remaining} vertices of the graph, \emph{from the perspective of $\PPi{k}$}, into a \emph{clique} entirely\footnote{We emphasize that this clique is not part of a single input graph, but rather is a \emph{union} of various inputs, which are all consistent with the view of player $\PPi{k}$ based on the input and messages received.}. Moreover, since we only removed $\lesssim n/k$ vertices for each player, 
	we still have $\approx n/k$ vertices left in this clique. Thus, the number of colors needed by $\PPi{k}$ to color this clique is $\approx n/k \gtrsim (\Delta/k^3)^{k}$ (which is larger than $\poly(\Delta)$ for sufficiently large $k$).  
\end{itemize}
\noindent
Finally, we  also state our compression lemma that is used to find the messages $M_1,\ldots,M_{k-1}$ that create ``small'' number of non-edges in the above discussion. 
\begin{itemize}[leftmargin=10pt]
	\item \textbf{Compression Lemma:} Let $H$ be any arbitrary graph and consider a distribution over subgraphs of $H$ obtained by sampling each edge with probability $p$. Any compression scheme 
	that maps the graphs sampled from this distribution into $s$-bit summaries will create a summary so that at most $O(s/p)$ edges are \emph{missing} from \emph{all} graphs mapped to this summary. 
\end{itemize}
This bound should be contrasted with more standard compression arguments that in the same setting, prove that $O(s \cdot \log^{-1}\!{(1/p)})$ \emph{edges} exist in \emph{all} graphs mapped to the summary. 
The proof is  a simple exercise in random graph theory plus showing that an $s$-bit compression cannot ``capture'' events that happen with probability $<2^{-s}$ in the input distribution. 
This concludes the description of our lower bound approach for establishing~\Cref{res:lower}. 

\subsection{Algorithms of~\Cref{res:upper}}

We now turn to our algorithmic results for multi-pass semi-streaming algorithms for graph coloring.

\paragraph{$O(\Delta^2)$ coloring in two passes.} The key ingredient of this algorithm is the following family of coloring functions for any integers $n,\Delta \geq 1$: 
\begin{itemize}
	\item $\CC(n,\Delta)$: there are $O(n)$ functions $C: V \rightarrow [\Delta]$ in the family so that given any $n$-vertex graph $G=(V,E)$ with max-degree $\Delta$, there is \emph{some} function $C$ in the family such that assigning color $C(v)$ to each vertex $v$ 
	only creates $O(n)$ monochromatic edges. Moreover, each of these functions can be implicitly stored in $O(\log{n})$ bits. 
\end{itemize}

The proof of existence of this family is via probabilistic method by choosing these functions to be near-universal hash functions and a simple probabilistic analysis. 

Now, consider the following simple two-pass algorithm. In the first pass, maintain $O(n)$ counters on the number of monochromatic edges of $G$ for each of the functions $C \in \CC(n,\Delta)$: the counter for function $C$ simply needs 
to add one for each edge $(u,v)$ appearing in the stream with $C(u) = C(v)$. This only requires $O(n)$ space. Given that we already know at least one of these counters only count up to $O(n)$ by the guarantee of 
$\CC(n,\Delta)$, we will use the function $C$ of that counter and \emph{store} all monochromatic edges of $G$ under $C$. Given that $G$ had maximum-degree $\Delta$, these monochromatic edges under $C$ can themselves 
be properly colored using $(\Delta+1)$ colors. Taking the \emph{product} of these two colorings then will give us an $O(\Delta^2)$ coloring as desired. 

\paragraph{$O(\Delta)$ coloring in $O(\log{\Delta})$ passes.} The idea behind this algorithm is to \emph{gradually} grow a coloring of $G$ over multiple passes, using an extension of the ideas in the previous algorithm. 
For this, we need another family of coloring functions for integers $n,\Delta$: 
\begin{itemize}
	\item $\CCstar(n,\Delta)$: there are $O(n)$ functions $C: V \rightarrow [O(\Delta)]$ in the family so that given any $n$-vertex graph $G=(V,E)$ with max-degree $\Delta$ and any partial (valid) coloring $C_0$ of some \emph{subset} of vertices, 
	there is \emph{some} function $C$ in the family such that assigning color $C(v)$ to every vertex $v$ \emph{uncolored} by $C_0$ only creates $o(n_0)$ monochromatic edges, where $n_0$ is the number of uncolored vertices by $C_0$. Moreover, each function can be implicitly stored in $O(\log{n})$ bits. 
\end{itemize}
The proof of existence of this family is again via probabilistic arguments although it needs a more detailed analysis. 

The algorithm is then as follows. We start with a coloring $C_0$ that leaves all vertices uncolored. Then, iteratively, we first make one pass and use $O(n)$ counters to find a desired coloring function $C \in \CCstar(n,\Delta)$ as specified by the above 
result; in the second pass we pick all $o(n_0)$ monochromatic edges of this coloring with respect to $C_0$. This allows us to color $(1-o(1))$ fraction of uncolored vertices of $C_0$ by $C$ without creating \emph{any} monochromatic edges. 
We continue this for $O(\log{\Delta})$ iterations so that $C_0$ only leaves $O(n/\Delta)$ vertices uncolored. We make one final pass over the input and store all $O(n)$ edges incident on these remaining vertices and then at the end, 
simply color them greedily using $(\Delta+1)$ colors (as \emph{any} partial coloring can be extended to a $(\Delta+1)$ coloring greedily). This gives our $O(\Delta)$ coloring algorithm. 

We conclude this part by noting that even though both our algorithms turn out quite simple, their design, based on families $\CC(n,\Delta)$ and $\CCstar(n,\Delta)$, requires a careful consideration to ensure one can also \emph{verify} the guarantees of respective families in \emph{limited space}\footnote{For instance, a ``more standard'' guarantee instead of $\CC(n,\Delta)$ that bounds the \emph{maximum-degree} of monochromatic edges can also be obtained via pair-wise independent hash functions (see, e.g.~\cite{CzumajDP20,BeraCG19}); but then that would 
require $\Theta(n)$ space \emph{per} each function to verify whether or not the function satisfies the desired property.}.


\section{Preliminaries}\label{sec:prelim}

\paragraph{Notation.}  For an integer $t \geq 1$, we define $[t] := \set{1,2,\ldots,t}$. For a tuple $(X_1,\ldots,X_t)$ and any $i \in [t]$, we define $X_{<i} := (X_1,\ldots,X_{i-1})$. For a distribution $\mu$, $\supp{\mu}$ denotes the support of $\mu$.

For a graph $G=(V,E)$, we use $\Delta(G)$ to denote the maximum degree of $G$. For a vertex $v \in V$, we use $N(v)$ to denote the neighbors of $v$ in $G$. For any integer $c \geq 1$, 
a \emph{$c$-coloring function} is \emph{any} function $C: V \rightarrow [c]$ and does \emph{not} necessarily need to be a proper coloring of $G$. A monochromatic 
edge under $C$ is any edge $(u,v)$ of $G$ with $C(u) = C(v)$.  
We further define a \emph{partial $c$-coloring function} as any function $C: V \rightarrow [c] \cup \set{\perp}$; we refer to vertices $v$ with $C(v) = \perp$ as uncolored vertices and \emph{not} consider edges 
$(u,v)$ with $C(u) = C(v) = \perp$ as monochromatic edges. 

\medskip

We use the following standard form of Chernoff bound in our proofs. 
 \begin{proposition}[Chernoff bound; c.f.~\cite{DubhashiP09}]\label{prop:chernoff}
 	Suppose $X_1,\ldots,X_m$ are $m$ independent random variables with range $[0,1]$ each. Let $X := \sum_{i=1}^m X_i$ and $\mu_L \leq \expect{X} \leq \mu_H$. Then, for any $\eps > 0$, 
 	\[
 	\Pr\paren{X >  (1+\eps) \cdot \mu_H} \leq \exp\paren{-\frac{\eps^2 \cdot \mu_H}{3+\eps}} \quad \textnormal{and} \quad \Pr\paren{X <  (1-\eps) \cdot \mu_L} \leq \exp\paren{-\frac{\eps^2 \cdot \mu_L}{2+\eps}}.
 	\]
 \end{proposition}
 
 \medskip
 
 Finally, we use the following basic property of any proper coloring, in creating many pairs of vertices which are colored the same. The proof is standard and is presented for completeness. 
\begin{proposition}\label{prop:C-coloring}
	In any proper $c$-coloring of a graph $G=(V,E)$ for $c \leq \frac{n}{2}$, there are at least $\dfrac{n^2}{4c}$ pairs of vertices that are colored the same. 
\end{proposition}
\begin{proof}
	For any $i \in [c]$, let $n_i$ vertices denote the number of vertices colored $i$ in the given $c$-coloring. 
	Since this is a proper coloring, we have, 
	\begin{align*}
		\text{number of pairs colored the same} &= \sum_{i=1}^{c} \binom{n_i}{2} = \frac{1}{2} \cdot \paren{\sum_{i=1}^{c} n_i^2 - n_i}\\ 
		&= \frac{1}{2} \cdot \paren{\sum_{i=1}^{c} n_i^2} - \frac{1}{2} \cdot n \tag{as $\sum_{i=1}^{c} n_i = n$ since all vertices are colored} \\
		&\geq \frac{1}{2} \cdot \paren{c \cdot (\frac{n}{c})^2 - n} \tag{as sum of quadratic-terms is minimized when they are all equal}  
		 \ge \frac{n^2}{4c},
	\end{align*}
	where the last inequality is by the assumption $c \leq n/2$. This concludes the proof. \Qed{prop:C-coloring} 
	
\end{proof}


\section{The Lower Bound}\label{sec:high-level}

We present our lower bound in this section and formalize~\Cref{res:lower}. We start by introducing a key tool used in our lower bound regarding a family of random graphs and its key compression aspect for our purpose. 
We then define the communication game we use in proving~\Cref{res:lower} formally, and next present the proof the communication lower bound.


\subsection{A Random Graph Distribution and its Compression Aspects}\label{sec:random}

We introduce a basic random graph distribution in this subsection that forms an important component of the analysis of our lower bound. 
The key difference of our distribution from standard random graph models is that it generates random subgraphs of \emph{arbitrary (base) graphs}, as opposed to subgraphs of cliques (which means some edges 
may never appear in the support of this distribution if they are not part of the base graph). The other change is that we ensure a \emph{deterministic} bound on the maximum degree of the graphs sampled from this distribution. 

\begin{Definition}\label{def:rg}
	For a \textbf{base graph} $\Gbase = (V,\Ebase)$ and parameters $p \in (0,1), d \geq 1$, we define the {\textbf{random graph distribution}} $\GR := \GR(\Gbase,p,d)$ as follows:
	\begin{enumerate}[label=$(\roman*)$,leftmargin=20pt]
		\item\label{line:rg1} Sample a graph $G$ on  vertices $V$ and edges $E$ by picking each edge of $\Ebase$ independently and with probability $p$ in $E$; 
		\item Return $G$ if $\Delta(G) < 2p \cdot d$, and otherwise repeat the process.
	\end{enumerate}
\end{Definition}

We will eventually set $d$ to be approximiately $\Delta(\Gbase)$. Since the average degree of a vertex is at most $p \cdot \Delta(\Gbase) \approx p \cdot d$, the second condition of~\Cref{def:rg} rarely kicks in by Chernoff bound, and thus this distribution is basically sampling random subgraphs of $\Gbase$. We will make these statements more precise in the proof of \Cref{clm:GRp}.

We now consider algorithms that aim to ``compress'' graphs sampled from $\GR$. 

\begin{Definition}\label{def:ca}
	Consider $\GR(\Gbase,p,d)$  for a base graph $\Gbase=(V,\Ebase)$ and parameters $p \in (0, 1), d \geq 1$, and an integer $s \geq 1$. 
	A {\textbf{compression algorithm}} with {{size}} $s$ is any function $\Phi: \supp{\GR} \rightarrow \set{0,1}^s$ that maps graphs sampled from $\GR$ 
	into $s$-bit strings. For any graph $G \in \supp{\GR}$, we refer to $\Phi(G)$ as the {\textbf{summary}} of $G$. For any summary $\phi \in \set{0,1}^s$, we define: 
	\begin{itemize}[leftmargin=20pt]
		\item $\GR_{\phi}$ as the distribution of graphs mapped to $\phi$ by $\Phi$, i.e., $\GR_{\phi} := G \sim \GR \mid \Phi(G) = \phi$. 
		\item $\Gmiss(\phi)=(V,\Emiss(\phi))$, called the {\textbf{missing graph}} of $\phi$, as a graph on vertices $V$ and edges missed by \emph{all} graphs in $\GR_{\phi}$, i.e., 
		\[
			\Emiss(\phi) := \set{(u,v) \in \Ebase \mid \text{no graph in $\supp{\GR_{\phi}}$ contains the edge $(u,v)$}}. 
		\]
	\end{itemize}
\end{Definition}

We use the graphs from distribution $\GR$ (for different base graphs and probability parameters) in the design of our lower bounds. The compression algorithms in~\Cref{def:ca} then correspond to streaming algorithms that compress these graphs into their $s$-bit memory.  

The notion of a missing graph is  particularly useful for us, as from the perspective the streaming algorithm, only pairs of vertices with an edge in the missing graph are known to \emph{not} have an edge in the original input. This implies that these are the only pairs of vertices that can be monochromatic in the final coloring without violating the correctness of the algorithm on \emph{some} input. 

The following lemma summarizes the main property of compression algorithms for our random graph distribution required in our main proof. Roughly speaking, it states that the missing graph of a ``small-size'' compression algorithm 
cannot have ``many'' edges\footnote{An intuition about the bounds of this lemma: given a graph $G$ with maximum degree $\Delta$, an $O(n\log{n})$-bit compression can  ensure the \emph{presence} of $O(n)$ edges in all graphs mapped 
to a fixed summary by storing these edges explicitly. However, it can also ensure \emph{absence} of up to $O(n^2/\Delta)$ edges from all graphs mapped to a summary by instead storing a $(\Delta+1)$ coloring of $G$. Our lower bound 
in this lemma focuses on the latter type of bounds and prove they are (almost) tight (for this specific instantiation, think of $s \approx n$, $p \approx \Delta/n$, and a clique for base graph).}.

\begin{lemma}\label{lem:rg}
	Let $\Gbase = (V,\Ebase)$ be an $n$-vertex graph, $s \geq 1$ be an integer, and $p \in (0,1)$ and $d \ge 1$ be parameters such that $d \ge \max\{\Delta(\Gbase), 4\ln{(2n)/p}\}$.  
	Consider the distribution $\GR := \GR(\Gbase,p,d)$ and 
	suppose $\Phi: \supp{\GR} \rightarrow \set{0,1}^s$ is a compression algorithm of size $s$ for $\GR$. Then, there exist a summary $\phistar \in \set{0,1}^s$ such that 
	in the missing graph of $\phistar$, we have 
	\[
		\card{\Emiss(\phistar)} \leq \frac{\ln{2} \cdot (s+1)}{p}. 
	\]
\end{lemma}
\begin{proof}
	Define the distribution $\GRp$ as the distribution of graphs in Line~\ref{line:rg1} of~\Cref{def:rg} (i.e., without the check on max-degree and re-sampling step). This way, we have 
	\begin{align}
		\GR = \paren{G \sim \GRp \mid \Delta(G) < 2p \cdot d}. \label{eq:rg1}
	\end{align}
	We shall use this view in the following for bounding the probabilities of certain events. In particular, we have the following simple claim. This is because the
	graphs sampled from $\GRp$ already satisfy the conditioned event above with high enough probability. 
	
	\begin{claim}\label{clm:GRp}
		For any event $\event$, 
		$
			\Pr_{\GR}\paren{\event} \leq 2 \cdot \Pr_{\GRp}\paren{\event}. 
		$
	\end{claim}
	\begin{proof}
		Fix any vertex $u \in V$, and for any $v \in N(u)$ (in $\Gbase$), define an indicator random variable $X_{uv} \in \set{0,1}$ which is $1$ iff the edge $(u,v)$ is sampled in $\GRp$. Let $X_u := \sum_{v \in N(u)} X_{uv}$
		which will be equal to the degree of $u$ in the sampled graph of $\GRp$. 
		
		We have that $\expect{X_u}$ is the expected degree of $u$, which is at most $p \cdot \Delta(\Gbase) \leq p \cdot d$ as 
		we sample each neighbor  of $u$ with probability $p$. Since $X_u$ is a sum of independent random variables, by Chernoff bound (\Cref{prop:chernoff} with $\varepsilon = 1$ and $\mu_H = p \cdot d$), 
		\[
			\Pr_{\GRp}\paren{X_u \geq 2p \cdot d} \leq \exp\paren{-\frac{p \cdot d}{4}} \leq \exp\paren{-\ln(2n)} = \frac{1}{2n},
		\]
		by the promise of the lemma statement that $d \ge 4 \ln(2n)/p$. A union bound on all $n$ vertices ensures that
		\[
		\Pr_{\GRp}\paren{\Delta(G) \geq 2p \cdot d} \leq n \cdot \frac{1}{2n} = \frac12.
		\]
		We can now conclude by~\Cref{eq:rg1} that for any event $\event$, 
		\begin{align*}
			\Pr_{\GR}\paren{\event} &= \Pr_{\GRp}\paren{\event \mid \Delta(G) < 2p \cdot d} = \frac{\Pr_{\GRp}\paren{\event ~\text{and}~ \Delta(G) < 2p \cdot d}}{\Pr_{\GRp}\paren{\Delta(G) < 2p \cdot d}} 
			\leq 2 \cdot \Pr_{\GRp}\paren{\event}, 
		\end{align*}
		as desired. \Qed{clm:GRp} 
		
	\end{proof}
	\noindent
	For any summary $\phi \in \set{0,1}^s$, its distribution $\GR_{\phi}$, and its missing graph $\Gmiss(\phi)$,  
	\begin{align}
		\Pr_{\GR}\paren{\text{$G$ is sampled from $\GR_{\phi}$}} \leq \Pr_{\GR}\paren{\text{no edge of $\Gmiss(\phi)$ is sampled in $G$}}, \label{eq:rg2} 
	\end{align}
	because edges in $\Gmiss(\phi)$ cannot belong to the graphs in the support of $\GR_{\phi}$ by~\Cref{def:ca}. We can bound the RHS of~\Cref{eq:rg2} 
	using the distribution $\GRp$ and apply~\Cref{clm:GRp} to get the result for $\GR$ also. By the independence in the choice of edges in $\GRp$, we have, 
	\begin{align*}
		\Pr_{\GRp}\paren{\text{no edge of $\Gmiss(\phi)$ is sampled in $G$}} = \hspace{-0.4cm}\prod_{e \in \Emiss(\phi)} \hspace{-0.4cm} (1-p) = (1-p)^{\card{\Emiss(\phi)}} \leq \exp\!\Paren{\!\!-p \cdot \card{\Emiss(\phi)}\!}. 
	\end{align*}
	Thus, combined by~\Cref{clm:GRp} and~\Cref{eq:rg2}, for any summary $\phi \in \set{0,1}^s$, we have, 
	\begin{align}
	\Pr_{\GR}\paren{\text{$G$ is sampled from $\GR_{\phi}$}} \leq 2 \cdot \exp\!\Paren{\!\!-p \cdot \card{\Emiss(\phi)}\!}. \label{eq:rg3} 
	\end{align}
	
	We now switch to lower bound the LHS of~\Cref{eq:rg3} instead. Since $\Phi$ maps each graph sampled from $\GR$ to one of $2^{s}$ messages $\phi \in \set{0,1}^s$, we have, 
	\[
		\sum_{\phi \in \set{0,1}^s} \Pr_{\GR}\paren{\text{$G$ is sampled from $\GR_{\phi}$}} = 1, 
	\] 
	which means that there exist some $\phistar \in \set{0,1}^s$ such that 
	\[
		\Pr_{\GR}\paren{\text{$G$ is sampled from $\GR_{\phistar}$}} \geq 2^{-s}. 
	\]
	Combining this with~\Cref{eq:rg3}, we have that 
	\[
		\exp\paren{-s \cdot \ln{2}} \leq \exp\!\Paren{\ln{2}-p \cdot \card{\Emiss(\phistar)}}, 
	\]
	which implies
	\[
		\card{\Emiss(\phistar)} \leq \frac{\ln{2} \cdot (s+1)}{p},
	\]
	concluding the proof. \Qed{lem:rg}
	
\end{proof}

\subsection{The Coloring Communication Game}\label{sec:cc}  
We prove our lower bound in~\Cref{res:lower} via communication complexity arguments in the following communication game. (The setting of this game is called the number-in-hand multi-party communication complexity with shared blackboard in the literature). 

\begin{Definition}\label{def:cc-coloring}
	For integers $n,\Delta,k \geq 1$, the $\coloring(n,\Delta,k)$ game is defined as:
	\begin{enumerate}[label=$\roman*).$, leftmargin=25pt]
		\item There are $k$ players $\PPi{1},\ldots,\PPi{k}$. Each player $\PPi{i}$ knows the vertex set $V$ and receives a set $\Ei{i}$ of edges. Letting $G = (V, E)$ where $E = \Ei{1} \sqcup \cdots \sqcup \Ei{k}$, players are guaranteed that on every input $\Delta(G) \leq \Delta$ and their goal is to output a proper coloring of $G$. 
		\item The communication is done using a shared blackboard. First player $\PPi{1}$ writes a message $\Mi{1}$ based on $\Ei{1}$ on the shared blackboard which will be visible to \emph{all} subsequent players. Then, 
		player $\PPi{2}$ writes the next message $\Mi{2}$ based on $\Ei{2}$	and $\Mi{1}$. The players continue like this until $\PPi{k}$ writes the last message $\Mi{k}$ which is a function of $\Ei{k}$ and $\Mi{<k}$. 
		\item The goal of the players is to output a valid coloring of the input graph $G$ by $\PPi{k}$ writing it last on the shared blackboard as the message $\Mi{k}$. 
	\end{enumerate}  
	The {\textbf{communication cost}} of a protocol used by the players is defined as the worst-case number of bits written by any one player on the blackboard on any input. 
\end{Definition}

The following proposition is standard. 

\begin{proposition}\label{prop:stream-cc}
	Suppose there is a function $f: \IN^+ \rightarrow \IN^+$ and a deterministic streaming algorithm that on any $n$-vertex graph $G$ with known maximum degree $\Delta$, 
	outputs an $f(\Delta)$-coloring of $G$ using $s=s(n,\Delta)$ bits of space. Then, there also exists a deterministic protocol for $\coloring(n,\Delta,k)$ for any $k > 1$ with communication cost $O(s)$ bits 
	that outputs an $f(\Delta)$-coloring of any input graph. 
\end{proposition}
\begin{proof}
	The players simply run the streaming algorithm on their input by writing the content of the memory of the algorithm from one player to the next on the blackboard, so that the next player can continue running the 
	algorithm on their input. At the end, the last player computes the output of the streaming algorithm and writes it on the blackboard. 
	
	The maximum message size written on the blackboard is proportional to the size of memory of the streaming algorithm and is thus $O(s)$ as desired. \Qed{prop:stream-cc}
	
\end{proof}

A careful reader may have noticed from~\Cref{prop:stream-cc} that in~\Cref{def:cc-coloring}, we do not even need the ability of the protocol to read the messages of \emph{all} prior players (via the blackboard), and the message of one player
to the next suffice. We allow for this extra power for technical reasons as it simplifies the proof of our lower bound (this is a typical approach in streaming lower bounds). 

The following is the main technical result of our paper. 

\begin{theorem}\label{thm:cc-lower}
	There are absolute constants $n_0,\eta_0 > 0$ such that the following is true. Consider any choice of the following parameters 
	\[
		n \geq n_0,\qquad  \Delta\geq 64\ln^2{(2n)}, \qquad  1 \leq k \leq \log_{\Delta}\!{(n)}, \qquad s \geq n\log{\Delta}.
	\]
	Then no deterministic communication protocol for $\coloring(n,\Delta,k)$ with communication cost $s$ can color every valid input graph with fewer than 
	\[
		\paren{\frac{1}{\eta_0 \cdot k}}^{2k} \cdot \paren{\frac{n \cdot \Delta}{s}}^{k} \text{colors}.
	\]
\end{theorem}
\noindent
As a corollary of this and~\Cref{prop:stream-cc}, we can formalize~\Cref{res:lower} as follows. (other settings of parameters in~\Cref{thm:cc-lower} imply various other
lower bounds for streaming algorithms also.) 

\begin{corollary}\label{cor:main-lower}
	For any  $q \geq 1$, $\alpha \in (0,1)$, and sufficiently large $n > 1$, 
	no deterministic single-pass streaming algorithm can obtain a proper coloring of every  graph with maximum degree at most $\Delta$ for the following parameters: 
	\begin{itemize}[leftmargin=20pt]
		\item[$i).$] $O(n\cdot \log^{q}{n})$ space and fewer than $\exp\!\paren{\Delta^{1/4q}}$ colors for $\Delta = 200\log^{q+1}\!{(n)}$;
		\item[$ii).$] $O(n^{1+\alpha})$ space and fewer than $\Delta^{{1}/{3\alpha}}$ colors for $\Delta = n^{2\alpha}$. 
	\end{itemize}
\end{corollary}
\begin{proof}
	We prove both parts by~\Cref{prop:stream-cc} and using different parameters in~\Cref{thm:cc-lower}. 
	\begin{itemize}[leftmargin=20pt]
		\item[$i).$] Set $k= \sqrt{\log_{\Delta}{n}} = \Theta(\sqrt{\frac{\log{n}}{\log\log{n}}})$ and $s=O(n\cdot \log^{q}{n})$. These parameters, plus $n$ and $\Delta$, satisfy the hypotheses of~\Cref{thm:cc-lower}. 
		As such, we get that the minimum number of colors needed to color the input graph in this case is at least 
		\begin{align*}
			\paren{\frac{1}{\eta_0 \cdot k}}^{2k} \hspace{-10pt}\cdot \paren{\frac{n \cdot \Delta}{s}}^{k} = \paren{\frac{\log\log{n} \cdot \log{n}}{\Theta(1) \cdot \log{n}}}^{\Theta(\sqrt{\frac{\log{n}}{\log\log{n}}})} > 
			\exp\paren{200 \cdot \sqrt{\frac{\log{n}}{\log\log{n}}}} \gg \exp\!\paren{\Delta^{1/4q}}, 
		\end{align*}
		by a simple calculation of the parameters in these bounds. 
		\item[$ii).$] Set $k=\log_\Delta{n} = 1/2\alpha$ and $s=O(n^{1+\alpha})$. These parameters, plus $n$ and $\Delta$, satisfy the hypotheses of~\Cref{thm:cc-lower}. 
		As such, we get that the minimum number of colors needed to color the input graph in this case is at least 
		\begin{align*}
			\paren{\frac{1}{\eta_0 \cdot k}}^{2k} \hspace{-10pt}\cdot \paren{\frac{n \cdot \Delta}{s}}^{k} = \paren{\frac{n^{\alpha}}{\Theta(1)}}^{1/2\alpha} = \Theta(\sqrt{n}) \gg \Delta^{1/3\alpha},
		\end{align*}
		again by a simple calculation. This concludes the proof. \Qed{cor:main-lower}
	\end{itemize}

\end{proof}

\subsection{A Communication Lower Bound for Coloring}\label{sec:lower}

Before getting to the lower bound construction, we specify a recursive set of parameters. 
\paragraph{Parameters.} Our construction is governed by the following two parameters:
\begin{itemize}
	\item $p_i$: the probability parameter used in defining the graph of each player $\PPi{i}$ from the random graph distribution $\GR$ for base graphs chosen by the adversary; 
	\item $d_i$: a threshold on maximum degree of base graph (used in $\GR$) chosen by the adversary for each player $\PPi{i}$. 
\end{itemize} 
These parameters are defined recursively as follows (these expression would become clear shortly from the description and analysis of the lower bound): 
\begin{align}
& d_1 = n, \quad p_1 := \frac{\Delta}{2k \cdot n}, \qquad \text{and for  $i > 1$:} \quad d_{i} = \frac{2 \, \ln{2} \cdot (s+1) \cdot 2k}{p_{i-1} \cdot n}, \quad p_i = \frac{\Delta}{2k \cdot d_i}\label{eq:parameters}.
\end{align}
It is easier for us to work with the recursive definitions of these parameters in most of the analysis (as their closed form is tedious to work with). But, we  also compute them explicitly as follows. 

\begin{claim}\label{clm:parameters}
	For any $i > 1$, we have, 
	\begin{align*}
		d_i &= n \cdot \paren{\frac{2\ln{2} \cdot (s+1) \cdot (2k)^2}{n \cdot \Delta}}^{i-1} \qquad p_i = \frac{\Delta}{2k \cdot n} \cdot \paren{\frac{n \cdot \Delta}{2\ln{2} \cdot (s+1) \cdot (2k)^2}}^{i-1}.
	\end{align*}
\end{claim}
\begin{proof}
	These equations can be verified by induction on $i \geq 1$ in~\Cref{eq:parameters}. \Qed{clm:parameters}
	
\end{proof}

The lower bound construction is as follows (see~\Cref{fig:contain} for an illustration). 

\begin{tbox}
	 An \textbf{adversary} that generates the ``hard'' input of players (using parameters in~\Cref{eq:parameters}).

	\begin{enumerate}[label=$(\roman*)$]
		\item Let $\Gbase(1)$ be a clique on $n$ vertices $V_1 = V$. 
		\item For $i=1$ to $k$:
		\begin{enumerate}[leftmargin=15pt]
			\item Let $\GR_i:= \GR(\Gbase(i),p_i,d_i)$ and let 
			\[
			\Phi_i = \Phi_i(\GR_i,\Mstari{<i}): \supp{\GR_i} \rightarrow \set{0,1}^s
			\]
			be the function generating the message of player $\PPi{i}$ after seeing messages $\Mstari{<i}$ of the first $i-1$ players; we ensure that the input graph of 
			player $\PPi{i}$ given previous messages $\Mstari{<i}$ is chosen from $\GR_i$ and thus this is well defined. 
			\item Notice that $\Phi_i$ is a compression algorithm. Apply \Cref{lem:rg} and let $\Mstari{i}$ be the special summary of this compression algorithm, i.e., message for $\PPi{i}$. We shall verify the hypotheses of the lemma in \Cref{lem:sat-hyp}.
			\item\label{line:induced} Let $V_{i+1}$ be the set of vertices in $\Gmiss(\Mstari{i})$ with degree at most $d_{i+1}$ and $\Gbase(i+1)$ be the subgraph 
			of $\Gmiss(\Mstari{i})$ induced on $V_{i+1}$. 
		\end{enumerate}
		\item Let $\Gi{i} = (V_{i}, \Ei{i})\in \supp{\GR_i}$ be such that $\Phi_i(\Gi{i}) = \Mstari{i}$ for all $i$. Give player $\PPi{i}$ the edge set $\Ei{i}$ as the adversarial input. We shall verify that this is a valid input in \Cref{lem:correct}. 
	\end{enumerate}
\end{tbox}

Notice that in this construction, we allow the players to know that their inputs come from a smaller distribution $\GR_i$, not the entire space of edges. This is convenient for our analysis, and since this only makes the players' jobs easier (as they
can simply ignore this information), this can only strengthen our lower bound.

It is  useful to note the containment relationships between various edge sets in the lower bound construction. For all $i \in [k]$, the edge sets $\Ei{i}$ and $\Emiss(\Mstari{i})$ are \emph{disjoint} because $\Gi{i}$ was mapped to $\Mstari{i}$, and both are subsets of $\Ebase(i)$ by definition. Also, $\Ebase(i)$ itself is a subset of $\Emiss(\Mstari{i-1})$ by Line~\ref{line:induced} of the construction, obtained by removing ``high degree'' vertices in $\Emiss(\Mstari{i-1})$. A visual is provided in~\Cref{fig:contain} for reference.

\begin{figure}
	\begin{center}
		\begin{tikzpicture}
		\draw[fill=lightgray] (0, 0) -- (10, 0) -- (10, 4) -- (0, 4) -- (0, 0) node[anchor=south west] {$\Ebase(1) = \binom{V}{2}$};
		\draw[fill=lightblue] (0.25, 0.75) -- (1.75, 0.75)  -- (1.75, 3.75) -- (0.25, 3.75) -- (0.25, 0.75) node[anchor=south west] {$\Ei{1}$};
		\draw[fill=lightred] (2, 0.75) -- (9.75, 0.75)  -- (9.75, 3.75) -- (2, 3.75) -- (2, 0.75) node[anchor=south west] {$\Emiss(\Mstari{1})$};
		\draw[fill=lightgray] (2.25, 1.5) -- (9.5, 1.5) -- (9.5, 3.5) -- (2.25, 3.5) -- (2.25, 1.5) node[anchor=south west] {$\Ebase(2)$};
		\draw[fill=lightblue] (2.5, 2.25) -- (4, 2.25) -- (4, 3.25) -- (2.5, 3.25) -- (2.5, 2.25) node[anchor=south west] {$\Ei{2}$};
		\draw[fill=lightred] (4.25, 2.25) -- (9.25, 2.25) -- (9.25, 3.25) -- (4.25, 3.25) -- (4.25, 2.25) node[anchor=south west] {$\Emiss(\Mstari{2})$};
		\node at (7.75, 2.75) {$\cdots$};
		\end{tikzpicture}
	\end{center}
	\caption{An illustration of edge set containments in the adversary construction.}
	\label{fig:contain}
\end{figure}
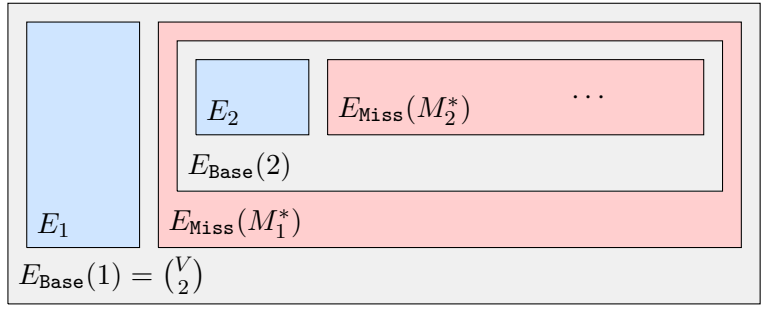

We start with two lemmas verifying that the above construction produces a valid input and satisfies the hypotheses of~\Cref{lem:rg} it invokes. 

\begin{lemma}\label{lem:sat-hyp}
	For all $i \in [k]$, the parameters $p_i$ and $d_i$ in the distribution $\GR_i = \GR(\Gbase(i), p_i, d_i)$ satisfy the hypotheses of \Cref{lem:rg}. That is, $d_i \ge \max\{\Delta(\Gbase(i)), 4\ln(2n)/p_i\}$ and $p_i \in (0,1)$.
\end{lemma}
\begin{proof}
	The fact that $d_i \ge \Delta(\Gbase(i))$ follows from $d_1 = n$ in the case $i = 1$, and directly from the construction of $\Gbase(i)$ in Line \ref{line:induced} for all other $i$. 
	
	To show that $d_i \ge 4\ln(2n)/p_i$,  we first note that $p_i \cdot d_i = \frac{\Delta}{2k}$ by definition of $p_i$ in~\Cref{eq:parameters}. Hence it suffices to show that $\Delta \ge 8k\ln(2n)$. Referencing the constraints in the statement of \Cref{thm:cc-lower}, we have  $\Delta \ge 64\ln^2(2n)$ which implies $\sqrt{\Delta} \ge 8\ln(2n)$, and we have $\sqrt{\Delta} \ge k$, which combined
	with the latter inequality, implies $\Delta \ge 8k\ln(2n)$. This proves the bound for $d_i$. 
	
	We now prove the bound for $p_i$. For this, it is easier to work with the closed-form of $p_i$ in~\Cref{clm:parameters}. We have, 
	\begin{align*}
		p_i &= \frac{\Delta}{2k \cdot n} \cdot \paren{\frac{n \cdot \Delta}{2\ln{2} \cdot (s+1) \cdot (2k)^2}}^{i-1} < \frac{\Delta^{i}}{n} \leq \frac{\Delta^k}{n} \leq 1, 
	\end{align*}
	as $s \geq n$ and $k > 1$ for the first inequality and by the upper bound of $k \leq \log_{\Delta}\!{(n)}$ for the last one. It is also clear that $p_i > 0$, 
	thus concluding the proof. \Qed{lem:sat-hyp}
	
\end{proof}

\begin{lemma}\label{lem:correct}
	Any graph $G$ constructed by the adversary has $\Delta(G) \leq \Delta$ and  no parallel edges. 
\end{lemma}
\begin{proof}
	Consider each graph $\Gi{i}$ as input to player $\PPi{i}$. We have, 
	\[
		\Delta(\Gi{i}) < 2p_i d_i = \frac{\Delta}{k}, 
	\]
	where the first inequality is by~\Cref{def:rg} for $\GR(\Gbase(i),p_i,d_i)$ and the second equality is by the definition of $p_i$ in~\Cref{eq:parameters}. This implies that the graph $\Gi{i}$ presented to each player has maximum degree at most $\Delta/k$. Given that there are $k$ players in the game, 
	this means the final graph has maximum degree at most $\Delta$. 
	
	To show that there are no parallel edges, simply note that $E_1, \dots, E_k$ are pairwise disjoint by the edge set containments noted above. \Qed{lem:correct}
	
\end{proof}

We  start proving the communication lower bound. First, we show that the set $V_{k+1}$ obtained at the end, i.e., after presenting last player's input, still is ``quite large''. 

\begin{lemma}\label{lem:set-size}
	For any $i \in [k+1]$, we have 
	$
		 \card{V_i} \geq n - (i-1) \cdot \dfrac{n}{2k}. 
	$
\end{lemma}
\begin{proof}
	The proof is by induction on $i$.  For $i=1$, we simply have $V_1 = V$ and thus $\card{V_1} = n$; hence, the base case holds. For the inductive step, it suffices to show that at most $\frac{n}{2k}$ vertices are removed after every player. By \Cref{lem:rg}, for which we verified the hypotheses in \Cref{lem:sat-hyp}, we have that $\Mstari{i}$ satisfies
	\[
		\card{\Emiss(\Mstari{i})} \leq \frac{\ln{2} \cdot (s+1)}{p_i}. 
	\]
	Recall that $V_{i+1}$ is the set of vertices with degree at most $d_{i+1}$ in $\Gmiss(\Mstari{i})$. Since any vertex in $V_{i} \setminus V_{i+1}$ contributes at least $d_{i+1}$ edges to $\Emiss(\Mstari{i})$ (and each edge can be 
	contributed at most twice), we have, 
	\[
		\card{\Emiss(\Mstari{i})} \geq \frac{1}{2} \cdot \card{V_i \setminus V_{i+1}} \cdot d_{i+1},
	\]
	implying that 
	\[
		\card{V_i \setminus V_{i+1}} \leq \frac{2 \, \ln{2} \cdot (s+1)}{p_i \cdot d_{i+1}} = \frac{n}{2k},
	\]
	by the choice of $p_i$ and $d_i$ in~\Cref{eq:parameters}. \Qed{lem:set-size}
	
\end{proof}

We now formalize the idea we alluded to after defining the missing graph in \Cref{def:ca}, where we described how only the edges appearing in the missing graph can have the same color assigned to both endpoints. Some extra care is needed here to account for the fact that the players have their own compression algorithm which is defined based on the messages of previous players. 

\begin{lemma}\label{lem:same-color}
	For any two vertices $u, v \in  V_{k+1}$ that have the same color in the output of $\PPi{k}$, the edge $(u,v)$ exists in $\Emiss(\Mstari{k})$.
\end{lemma}
\begin{proof}
Suppose toward a  contradiction that $(u, v) \not \in  \Emiss(\Mstari{k})$. We first show that there exists $i$ such that $(u, v) \not \in \Emiss(\Mstari{i})$ and $(u, v) \in \Ebase(i)$. Recalling that $\Emiss(\Mstari{k}) \subseteq \dots \subseteq \Emiss(\Mstari{1})$, either $(u, v) \not \in \Emiss(\Mstari{1})$, in which case taking $i=1$ suffices, or $(u, v) \in \Emiss(\Mstari{i-1}) \setminus \Emiss(\Mstari{i})$ for some $i > 1$. Referencing the containments illustrated in~\Cref{fig:contain}, either $(u, v) \in \Ebase(i)$ or $(u, v) \in \Emiss(\Mstari{i-1}) \setminus \Ebase(i)$. By Line~\ref{line:induced} of the construction, the second case happens only when $u$ or $v$ is dropped when restricting to $V_{i}$, which is impossible because $u, v \in V_{k+1} \subseteq V_{i}$, so $(u, v) \in \Ebase(i)$ as desired. 

Because $(u, v) \not \in \Emiss(\Mstari{i})$ and $(u, v) \in \Ebase(i)$, there should exists some graph $G'_i \in \supp{\GR_i}$ that contains the edge $(u,v)$ and is mapped to $\Mstari{i}$ by player $\PPi{i}$, i.e., $\Phi_i(G'_i) = \Mstari{i}$. 
Consider giving the graphs $\Gi{1}, \dots, \Gi{i-1}, G'_i, \Gi{i+1}, \dots, \Gi{k}$ as input to the players $\PPi{1}, \dots, \PPi{k}$, respectively. Because the same messages are generated as in the original construction, $\PPi{k}$ also outputs the same coloring. But now $(u, v)$ is in the input graph, so $u$ and $v$ should be colored differently.  \Qed{lem:same-color}

\end{proof}

In this fina lemma, we bound the  number of colors that can be used by player $\PPi{k}$ to color $G$.

\begin{lemma}\label{lem:color-size}
	Player $\PPi{k}$ requires
	\[
		c \geq \frac{n^2}{16\ln{2} \cdot (s+1)} \cdot p_k
	\]
	colors to color the graph $G$ constructed by the adversary above.
\end{lemma}
\begin{proof}
	Consider the number of pairs of vertices in $V_{k+1}$ that are assigned the same color by the proper $c$-coloring created by player $P_k$. Because $V_{k+1}$ has at least $\frac{n}{2}$ vertices by \Cref{lem:set-size}, the number of pairs is at least $\frac{n^2}{16c}$ by \Cref{prop:C-coloring}. (We have $c \le \frac{n}{2}$ by the choice of $s > n$.) At the same time, the number of pairs of vertices that can be colored the same is at most $|\Emiss(\Mstari{k})|$ by \Cref{lem:same-color}, which by \Cref{lem:rg} is at most $\frac{\ln 2 \cdot (s + 1)}{p_k}$. In conclusion,
	\begin{equation*}
		\frac{n^2}{16c} \le \frac{\ln 2 \cdot (s + 1)}{p_k},
	\end{equation*}
	which rearranges to our desired bound. \Qed{lem:color-size} 
	 	
\end{proof}

Finally, by plugging in the explicit value of $p_k$ in~\Cref{clm:parameters} in the bounds of~\Cref{lem:color-size}, we have that the minimum number of colors $c$ used by the protocol is at least 
\begin{align*}
	c &\geq \frac{n^2}{16\ln{2} \cdot (s+1)} \cdot \frac{\Delta}{2k \cdot n} \cdot \paren{\frac{n \cdot \Delta}{2\ln{2} \cdot (s+1) \cdot (2k)^2}}^{(k-1)} \\
	&= \frac{k}{4} \cdot \paren{\frac{n \cdot \Delta}{2\ln{2} \cdot (s+1) \cdot (2k)^2}}^{k} \\
	&\geq \paren{\frac{1}{\eta_0 \cdot k}}^{2k} \cdot \paren{\frac{n \cdot \Delta}{s}}^{k}, 
\end{align*}
for some absolute constant $\eta_0 < 100$. This concludes the proof of~\Cref{thm:cc-lower}.


\section{The Algorithms}\label{sec:alg}

We present our algorithmic results in this section that complement our strong lower bound for single-pass algorithms. Our first algorithm achieves $O(\Delta^2)$-coloring in only two passes. 

\begin{theorem}\label{thm:alg-2pass}
	There exists a deterministic algorithm that given any $n$-vertex graph $G$ with maximum degree $\Delta$ presented in an insertion-only stream, can find an $O(\Delta^2)$-coloring of $G$ in two passes and $O(n\log{n})$ bits of space. 
\end{theorem}

Our second algorithm builds on the ideas developed for the first one and reduces the number of colors to $O(\Delta)$, at the cost of increasing the number of passes to $O(\log{\Delta})$. 

\begin{theorem}\label{thm:alg-log-pass}
	There exists a deterministic algorithm that given any $n$-vertex graph $G$ with maximum degree $\Delta$ presented in an insertion-only stream, can find an $O(\Delta)$-coloring of $G$ in $O(\log{\Delta})$ passes and $O(n\log{n})$ bits of space. 
\end{theorem}

In the following, we first present two families of coloring functions that create few monochromatic edges in different settings, needed for our algorithms, and then present each of our algorithms. Further extensions of our results 
such as to dynamic streams are presented at the end of this section. These results collectively formalize~\Cref{res:upper}. 

\subsection{Families of Coloring Functions with Few Monochromatic Edges} 

We start with the following simple result  that shows existence of a fixed family of $\Delta$-coloring functions that allows for coloring any graph $G$ with $O(n)$ monochromatic edges via at least one of the functions in the family. 
We shall use this result in our two-pass algorithm. 
\begin{lemma}\label{lem:mono-coloring}
	For any integers $n,\Delta \geq 1$, there exists a family $\CC := \CC(n,\Delta)$ of size  at most $(2n)$ consisting of $\Delta$-coloring functions such that for any $n$-vertex graph $G=(V,E)$ with maximum degree $\Delta$, 
	there is a coloring function $C \in \CC$ such that $G$ has at most $(4n)$ monochromatic edges under $C$. Moreover, each  function in $\CC$ can be generated via $O(\log{n})$ bits. 
\end{lemma}

\begin{proof}
The proof is by a probabilistic method. Let $p$ be the smallest prime number larger than $n$ and note that we have $p < 2n$ by Bertrand's postulate. We simply pick $\CC$ to be the 
following standard family of near-universal hash functions: 
\begin{align*}
	\CC := \set{C_{a}(v) = ((a \cdot v \hspace{-0.3cm} \mod p)\hspace{-0.3cm} \mod \Delta)+1~\text{for all $v \in V$} \mid a \in \set{0,1,\ldots,p-1}}. 
\end{align*} 
As such, since $\CC$  is a  near-universal hash family, for any two vertices $u,v \in V$, we have, 
\begin{align}
	\Pr_{C \in \CC}\paren{C(u) = C(v)} \leq \frac{2}{\Delta}. \label{eq:sc2} 
\end{align}
For completeness, we provide the standard argument that proves~\Cref{eq:sc2}. Fix any two vertices $u \neq v \in V$ and consider  $C_a \in \CC$. For $C_a(u)$ to be equal to $C_a(v)$, we should have, 
	\[
		(a \cdot (u-v) \hspace{-0.3cm} \mod p) \in \set{-\floor{\frac{(p-1)}{\Delta}} \cdot \Delta,\ldots,-2\Delta,-\Delta, 0 , \Delta, 2\Delta, \ldots, \floor{\frac{(p-1)}{\Delta}} \cdot \Delta}, 
	\]
	which includes at most $2p/\Delta$ choices in the RHS. Since $p$ is a prime, for any number $z$ in the RHS, there is only a unique choice of $a \in \set{0,1,\ldots,p-1}$ that can result in $a \cdot (u-v)$ to be equal to $z \!\!\mod p$. 
	As such, for a random $C_a$, the probability that $C_a(u) = C_a(v)$ is at most $2/\Delta$ as desired.

Using~\Cref{eq:sc2}, for any graph $G$, we have, 
\[
	\Exp_{C \in \CC}\Bracket{\text{\# of monochromatic edges of $G$ under $C$}} = \sum_{(u,v) \in E} \Pr_{C \in \CC}\paren{C(u) = C(v)} \leq 2n\Delta \cdot \frac{2}{\Delta} = 4n. 
\]
Consequently, for any given graph $G$, there should exist a choice of $C \in \CC$ with at most $(4n)$ monochromatic edges. Finally, any coloring function in $\CC$ is specified uniquely by an integer in $\set{0,1,\ldots,p-1}$ which requires
$O(\log{n})$ bits to store. This concludes the proof. 
\Qed{lem:mono-coloring}

\end{proof}

We next present our second family of  functions which is used in our $O(\Delta)$ coloring algorithm. 

\begin{Definition}\label{def:extensions}
	Let $C_1: V \rightarrow [c] \cup \set{\perp}$ be a partial $c$-coloring function of a graph $G=(V,E)$ that has no monochromatic edges. 
	Let $C_2: V \rightarrow [c]$ be a $c$-coloring function of $V$ (not necessarily a proper one). 
	We define the \textbf{extension of $C_1$ by $C_2$} as the $c$-coloring function $C_3: V \rightarrow [c]$ such that for any $v \in V$,  
	\[	
		C_3(v) = \begin{cases} C_1(v) & \text{if $C_1(v) \neq \perp$} \\ C_2(v) & \text{otherwise} \end{cases},
	\]
	i.e., $C_3$ uses $C_1$ to color  vertices $v$ with $C_1(v) \neq \perp$ and use $C_2$ to color the remaining vertices.
\end{Definition}

The following family of coloring functions has the following property: for any graph $G$ and a partial coloring $C_1$ of $G$, there is a coloring function $C$ in the family 
with a small number of monochromatic edges in the extension of $C_1$ by $C$. Formally, 
\begin{lemma}\label{lem:mono-partial-coloring}
	For any integers $n,\Delta \geq 1$, there exists a family $\CCstar := \CCstar(n,\Delta)$ of size  at most $(2n)$ consisting of $(6\Delta)$-coloring functions such that the following is true. 
	For any $n$-vertex graph $G=(V,E)$ with maximum degree $\Delta$ and any partial coloring function $C_1$ of $G$ with no monochromatic edges, 
	there is a $(6\Delta)$-coloring function $C \in \CCstar$ such that extension of $C_1$ by $C$ has at most $(\dfrac{n_0}{3})$ monochromatic edges where $n_0 := \card{\set{v \in V \mid C_1(v) = \perp}}$, 
	is the number of uncolored vertices by $C_1$. 
	Moreover, each  function in $\CCstar$ can be generated via $O(\log{n})$ bits. 
\end{lemma}
\begin{proof}
	The proof is again by the probabilistic method similar to that of~\Cref{lem:mono-coloring}. Let $p$ be the smallest prime number larger than $n$ and note that we have $p < 2n$ by Bertrand's postulate. We  pick $\CCstar$ to be the 
following standard family of near-universal hash functions: 
\begin{align*}
	\CCstar := \set{C_{a}(v) = ((a \cdot v \hspace{-0.3cm} \mod p)\hspace{-0.3cm} \mod 6\Delta)+1~\text{for all $v \in V$} \mid a \in \set{0,1,\ldots,p-1}}. 
\end{align*} 
Since $\CCstar$  is a  near-universal hash family, for any two vertices $u,v \in V$ and any fixed color $c \in [6\Delta]$, \begin{align}
	\Pr_{C_2 \in \CCstar}\paren{C_2(u) = C_2(v)} \leq \frac{2}{6\Delta} = \frac{1}{3\Delta} \qquad \text{and} \qquad 
	\Pr_{C_2 \in \CCstar}\paren{C_2(u) = c} \leq \frac{2}{6\Delta} = \frac{1}{3\Delta}. \label{eq:scstar2} 
\end{align}
The proof is identical to that of~\Cref{eq:sc2} by taking into account that the range of functions in $\CCstar$ is now $[6\Delta]$. We thus omit the proof.

For any edge $(u,v) \in E$ to be monochromatic in the \emph{extension} $C_3$ of $C_1$ by $C_2$, we should have that at least one of $C_1(u)$ or $C_1(v)$ is $\perp$; otherwise, both retain $u$ and $v$ their colors in $C_1$ which contains no monochromatic 
edges. By symmetry suppose $C_1(u) = \perp$ and so $u$ will be colored by $C_2$ in the extension $C_3$. If $C_1(v) = \perp$ also, then to get a monochromatic edge, we need $C_2(u) = C_2(v)$ which happens with probability at most $1/3\Delta$ 
by the first part of~\Cref{eq:scstar2}. Conversely, if $C_1(v) \neq \perp$, then to get a monochromatic edge, we need $C_2(u) = C_1(v)$ which again happens with probability at most $1/3\Delta$ by the second part of~\Cref{eq:scstar2}. 
All in all, only edges incident on $\set{v \in V \mid C_1(v) = \perp}$ can be monochromatic and each one will become so with probability at most $1/3\Delta$. Hence, 
\begin{align*}
	\Exp_{C_2 \in \CCstar}\Bracket{\text{\# of monochromatic edges of $G$ in extension of $C_1$ by $C_2$}} &\leq \sum_{v: C_1(v)=\perp} \sum_{u \in N(v)} \frac{1}{3\Delta} \\
	&= n_0 \cdot \Delta \cdot \frac{1}{3\Delta} = \frac{n_0}{3}. 
\end{align*}
 Hence, for any $G$ and $C_1$, there should exist a choice of $C_2 \in \CC$ with at most $(n_0/3)$ monochromatic edges in the extension of $C_1$ by $C_2$. 
 Also, any coloring function in $\CC$ is specified uniquely by an integer in $\set{0,1,\ldots,p-1}$ which requires $O(\log{n})$ bits to store, concluding the proof. 
\Qed{lem:mono-partial-coloring}

\end{proof}

\subsection{A Two-Pass $O(\Delta^2)$-Coloring Algorithm}\label{sec:2-pass} 

We  now present our two-pass semi-streaming algorithm for $O(\Delta^2)$ coloring and prove~\Cref{thm:alg-2pass}. The key tool we use in this result is the coloring functions of~\Cref{lem:mono-coloring}. 

\clearpage

\begin{Algorithm}\label{alg:2pass}
{A two-pass deterministic semi-streaming algorithm for $O(\Delta^2)$ coloring.} 

\begin{enumerate}[label=$(\roman*)$]
\item Let $\CC = \CC(n,\Delta) = \set{C_1,\ldots,C_k}$ be the family of $\Delta$-coloring functions guaranteed by~\Cref{lem:mono-coloring} for some $k \leq 2n$. 
\item In the \underline{first pass}, for any $i \in [k]$, maintain a \emph{counter} $\phi_i$ that counts the number of monochromatic edges of $G$ under the coloring $C_i$, i.e., 
\[
\phi_i = \card{\set{(u,v) \in G \mid C_i(u) = C_i(v)}}.
\]
Let $C_{\istar} \in \CC$ be the coloring function with the smallest value of $\phi_{\istar}$, i.e., $\istar \in \arg\min_{i \in [k]} \phi_i$. 
\item In the \underline{second pass}, store all monochromatic edges of $G$ under $C_{\istar}$. Compute a $(\Delta+1)$ coloring $C$ of the stored edges and return the following coloring function $C^{\star}$ 
as the answer: 
\[
	\text{for all $v \in V$: $C^{\star}(v) = (C_{\istar}(v)-1) \cdot (\Delta +1) + C(v)$}.
\]
\end{enumerate}
\end{Algorithm}

\begin{lemma}\label{lem:alg2-space}
The space complexity of~\Cref{alg:2pass} is $O(n\log{n})$ bits. 
\end{lemma}
\begin{proof}
The first pass of this algorithm requires storing $O(n)$ counters of size $O(\log{n})$ bits each, and can be implemented in $O(n\log{n})$ bits of space. 
The second pass requires storing only $O(n)$ edges by the guarantee of~\Cref{lem:mono-coloring} which again can be done in $O(n\log{n})$ bits of space. \Qed{lem:alg2-space}

\end{proof}

We now argue that the final coloring $C^{\star}$ returned by the algorithm is a proper coloring of $G$, i.e., it does not contain any monochromatic edges. 

\begin{lemma}\label{lem:alg2-correct}
\Cref{alg:2pass} always outputs a proper $O(\Delta^2)$ coloring of any given input graph with maximum degree $\Delta$.
\end{lemma}
\begin{proof}
Firstly, since maximum degree of $G$ is $\Delta$, we clearly have 
that maximum degree of stored edges is also at most $\Delta$, and consequently, the algorithm can always find a $(\Delta+1)$ coloring of the stored edges. For any edge $(u,v) \in G$, if $C_{\istar}(u) \neq C_{\istar}(v)$, 
\[
\card{C^{\star}(u) - C^{\star}(v)} \geq \card{C_{\istar}(u)-C_{\istar}(v)} \cdot (\Delta+1) -  \card{C(u)-C(v)} \geq (\Delta+1) - \Delta =1,
\]
thus $C^{\star}(u) \neq C^{\star}(v)$ and so $(u,v)$ will not be monochromatic. For any edge $(u,v) \in G$ with $C_{\istar}(u) = C_{\istar}(v)$, the algorithm stores $(u,v)$ in the second pass 
and thus by the coloring it finds, we have $C(u) \neq C(v)$, making $C^{\star}(u) \neq C^{\star}(v)$ also. 

Finally, since the total number of colors used by $C^{\star}$ is $\Delta \cdot (\Delta+1)$, we obtain an $O(\Delta^2)$ coloring as desired. \Qed{lem:alg2-correct}

\end{proof}

This concludes the proof of~\Cref{thm:alg-2pass}. 

\subsection{An $O(\log{\Delta})$-Pass $O(\Delta)$-Coloring Algorithm}\label{sec:log-pass}

This section includes our $O(\log{\Delta})$-pass semi-streaming algorithm for $O(\Delta)$ coloring, i.e., the proof of~\Cref{thm:alg-log-pass}. The key tool we use in this result is the coloring functions of~\Cref{lem:mono-partial-coloring}. 

\begin{Algorithm}\label{alg:log-pass}
{An $O(\log{\Delta})$-pass deterministic semi-streaming algorithm for $(6\Delta)$ coloring.} 

\begin{enumerate}[label=$(\roman*)$]
\item Let $\CCstar = \CCstar(n,\Delta) = \set{C_1,\ldots,C_k}$ be the family of $(6\Delta)$-coloring functions guaranteed by~\Cref{lem:mono-partial-coloring} for some $k \leq 2n$. 
\item Let $C$ be a partial coloring function, initially set to map all vertices to $\perp$. 
\item While $C$ has more than $n/\Delta$ uncolored vertices: 
\begin{enumerate}
\item In \underline{one pass}, for any $i \in [k]$, maintain a \emph{counter} $\phi_i$ that counts the number of monochromatic edges of $G$ under the extension $C'_i$ of $C$ by $C_i$, i.e., 
\[
\phi_i = \card{\set{(u,v) \in G \mid C'_i(u) = C'_i(v)}}.
\]
Let $C_{\istar} \in \CC$ be the coloring function with the smallest $\phi_{\istar}$, i.e., $\istar \in \arg\min_{i \in [k]} \phi_i$ and $C'_{\istar}$ be the extension of $C$ by $C_{\istar}$. 
\item In \underline{another pass}, store all monochromatic edges of $G$ under $C'_{\istar}$. For any vertex $v \in V$, if no monochromatic edges incident on $v$ are stored, then set $C(v) = C'_{\istar}(v)$. 
\end{enumerate}
\item Store all edges incident on the uncolored vertices of $C$. Greedily color all the remaining uncolored vertices with a color not assigned to their neighbors. 
\end{enumerate}
\end{Algorithm}

We first note a direct invariant of the algorithm that will be used in our analysis. 
\begin{lemma}\label{lem:log-pass-invariant}
At any point of time in~\Cref{alg:log-pass}, there are no monochromatic edges between vertices colored by $C$.
\end{lemma}
\begin{proof}
This is simply because we always work with the extensions of $C$ and thus if a vertex is colored by $C$, 
we never change its color, and since we only color a vertex by $C$ if it does not have any monochromatic edges. \Qed{lem:log-pass-invariant}

\end{proof}

Note that \emph{if} the while-loop finishes, then the coloring $C$ computed greedily by the algorithm is a proper $(6\Delta)$ coloring of $G$ as $C$ contained no monochromatic edges throughout (by~\Cref{lem:log-pass-invariant}, and the last step of using greedy coloring, only requires $(\Delta+1)$ colors since we have stored all edges incident on uncolored vertices. We thus want to show that the while-loop indeed finishes. This is the main part of the analysis. 

\begin{lemma}\label{lem:log-pass-finish}
	There are $O(\log{\Delta})$ iterations of the while-loop in~\Cref{alg:log-pass} before it terminates. 
\end{lemma}
\begin{proof}
	Fix an iteration of the while-loop and 
let $n_0 := \card{\set{v \in V \mid C(v) = \perp}}$ denote the number of uncolored vertices by $C$ at the beginning of this iteration. By the guarantee of~\Cref{lem:mono-partial-coloring} (and since~\Cref{lem:log-pass-invariant} verifies the hypothesis of this lemma), we know that 
the coloring $C'_{\istar}$ computed by the algorithm in this iteration at most $n_0/3$ monochromatic edges. This means that at least $n_0 - 2n_0/3 = n_0/3$ vertices not colored by $C$ have zero monochromatic 
edges under $C'_{\istar}$. All these vertices will now be colored by $C$ at the end of this iteration. 

By the above discussion, the number of uncolored vertices reduces by a factor of at most $2/3$ in each iteration. 
As a result, after $O(\log{\Delta})$ iterations, the number of uncolored vertices by $C$ drops below $n/\Delta$ and thus the while-loop terminates. \Qed{lem:log-pass-finish}

\end{proof}

Finally, we analyze the space complexity of the algorithm. 
\begin{lemma}\label{lem:log-pass-space}
The space complexity of~\Cref{alg:log-pass} is $O(n\log{n})$ bits. 
\end{lemma}
\begin{proof}
The first pass of each iteration of while-loop of~\Cref{alg:log-pass} requires maintaining $O(n)$ counters of size $O(\log{n})$ bits each, and can be implemented in $O(n\log{n})$ bits of space. 
The second pass requires storing only $O(n)$ edges by the guarantee of~\Cref{lem:mono-partial-coloring} (and since~\Cref{lem:log-pass-invariant} verifies the hypothesis of this lemma) which again can be done in $O(n\log{n})$ bits of space. 
Finally, at the end we are storing at most $\Delta$ edges for each of the remaining $n/\Delta$ uncolored vertices and thus we can store them in $O(n\log{n})$ bits as well. \Qed{lem:log-pass-space}

\end{proof}

This concludes the proof of~\Cref{thm:alg-log-pass}.

\subsection{Further Extensions}\label{sec:extensions}

\paragraph{Dynamic streams.} In order to implement our algorithms in dynamic streams, we simply need a way of recovering the $O(n)$ monochromatic edges in each step of each one. (Maintaining the counters is straightforward by simply adding and subtracting their values based on insertion and deletion of monochromatic edges -- recall that we already know the coloring we need to work with and thus upon update of an edge, we know whether or not it is a monochromatic edge).

To recover these $O(n)$ monochromatic edges, we can simply use any standard \emph{deterministic sparse recovery} algorithm over dynamic streams. The following result is folklore. 

\begin{proposition}[Folklore]\label{prop:sparse-recovery}
	There exists a deterministic algorithm that given an integer $k \geq 1$ and a dynamic stream of of edge insertions and deletions for an $n$-vertex graph $G$, uses $O(k \cdot \log{n})$ bits of space
	and at the end of the stream recovers all edges of in the graph \emph{under the promise} that $G$ has at most $k$ edges (the answer can be arbitrary when the promise is not satisfied).   
\end{proposition}

As in our algorithms we only need to find monochromatic edges that are guaranteed to be at most $O(n)$ many, we can simply use~\Cref{prop:sparse-recovery} to recover these edges in $O(n\log{n})$ bits even 
in dynamic streams (we simply need to define the underlying graph as insertions and deletions \emph{between} monochromatic pairs and set $k=O(n)$). 

This immediately extends both our~\Cref{thm:alg-2pass,thm:alg-log-pass} to dynamic streams with the same asymptotic space complexity and the same exact number of passes.

\paragraph{Removing the knowledge of $\Delta$.} Our algorithms in the previous part are described assuming the knowledge of $\Delta$. For our $O(\Delta)$ coloring algorithm this is simply without loss of generality as we can increase the 
number of passes by one and compute $\Delta$ in the first pass---given that we report the number of passes asymptotically anyway, this does not change anything. But the same approach for our $O(\Delta^2)$ coloring
algorithm  increases the number of passes to three instead. 

Nevertheless, there is a simple way to fix $O(\Delta^2)$ coloring algorithm without changing the number of passes. In the first pass, pick $O(\log{n})$ choices of for $\Delta$ in geometrically increasing values and maintain the 
counters for $C(n,\cdot)$ for these $O(\log{n})$ choices; in parallel, also compute $\Delta$ in this pass. At the end of the first pass, we  know $\Delta$ and can focus on the right choice of counters for $C(n,\Delta')$ where $\Delta' \geq \Delta \geq \frac12\cdot \Delta'$. The rest of the algorithm and its proof are exactly as before.

\subsection*{Acknowledgement} 

We thank the organizers of DIMACS REU in Summers 2020 and 2021, in particular Lazaros Gallos, for making this collaboration possible and all their help and encouragements along the way.   

\bibliographystyle{alpha}
\bibliography{new}

\end{document}